\documentclass[final,11pt]{article}
\usepackage[colorlinks=true,linkcolor=blue,citecolor=magenta]{hyperref}
\usepackage{amssymb,amsmath,amsthm}
\usepackage{cite}
\usepackage[utf8]{inputenc}
\usepackage{enumerate}
\usepackage{multicol}
\usepackage{lmodern}
\usepackage{microtype}
\usepackage{multirow}%
\usepackage{amsthm}%
\usepackage{mathrsfs}%
\usepackage[title]{appendix}%
\usepackage{xcolor}%
\usepackage{textcomp}%
\usepackage{manyfoot}%
\usepackage{booktabs}%
\usepackage{algorithm}%
\usepackage{algorithmicx}%
\usepackage{algpseudocode}%
\usepackage{listings}
\usepackage{colonequals}
\usepackage{subcaption}
\usepackage{dsfont}%
\usepackage{graphicx}
\usepackage[conditional,light,first,bottomafter]{draftcopy}
\draftcopyName{DRAFT\space\today}{130}
\draftcopySetScale{65}
\usepackage[letterpaper,hmargin=2.5cm,vmargin=2.8cm]{geometry}
\usepackage{setspace}
\makeatletter
\renewcommand{\section}{\@startsection%
{section}%
{1}%
{0em}%
{1.7em}%
{1.2em}%
{\normalfont\large\centering\bfseries}}
\renewcommand{\@seccntformat}[1]%
{\csname the#1\endcsname.\hspace{0.5em}}
\makeatother

\numberwithin{equation}{section}
\newtheorem{theorem}{Theorem}[section]
\newtheorem{proposition}[theorem]{Proposition}

\theoremstyle{definition}
\newtheorem{definition}[theorem]{Definition}
\newtheorem{properties}[theorem]{Properties}

\makeatletter
\newcommand{\mathleft}{\@fleqntrue\@mathmargin0pt}
\newcommand{\mathcenter}{\@fleqnfalse}
\makeatother

\begin{document}
\begin{titlepage}
\title{Uniqueness of star central configurations in the $5$-body problem
\footnotetext{%
Mathematics Subject Classification(2010):
70F10, 70F15
}
\footnotetext{%
Keywords:
Central Configuration, Regular polygon configuration, $5$-body problem, Uniqueness
}
\hspace{-8mm}
}
\author{
\textbf{Leasly A. Campa-Raymundo}
\\
Departamento de Matem\'aticas\\
Universidad Aut\'onoma Metropolitana, Iztapalapa Campus\\
San Rafael Atlixco 186, Iztapalapa\\
C.P. 09340, Ciudad de M\'exico\\
\texttt{mathcampa@gmail.com}
\\[4mm]
\textbf{Luis Franco-P\'erez}
\\
Departamento de Matem\'aticas Aplicadas y Sistemas\\
Universidad Aut\'onoma Metropolitana - Cuajimalpa\\
Av. Vasco de Quiroga 4871
C.P. 05348, Ciudad de M\'exico\\
\texttt{lfranco@cua.uam.mx}
}
\date{}
\maketitle
\vspace{-4mm}
\begin{center}
\begin{minipage}{5in}
  \centerline{{\bf Abstract}} \bigskip
\large
In this study, we present a rigorous analytical proof of the uniqueness of central configurations for the five-body problem, assuming that all five masses are equal and positioned at the vertices of a planar polygon. We consider configurations in which the bodies are equally spaced in angular position relative to the center of mass, and aim to determine whether a central configuration arises under these constraints. 
We prove that the only central configuration that satisfies these conditions occurs when the five bodies form a regular pentagon. Our approach is entirely analytical, relying on algebraic techniques rather than numerical approximations. By transforming the governing equations into a reduced system involving only two variables, we analyze the solution space over a significant and carefully bounded domain. This domain is divided into sixteen disjoint regions, within which we rule out additional solutions through explicit algebraic arguments.
Our results confirm that the regular pentagonal configuration is the only central configuration in this symmetric five-body scenario.
\normalsize
\end{minipage}
\end{center}
\thispagestyle{empty}
\end{titlepage}

\section{Introduction}\label{sec1}

The $n$-body problem is an important part of Celestial Mechanics. It involves studying the movement of $n$ point masses that interact only through universal gravitational forces. The case where $n=2$ has been solved, with the two bodies following conic trajectories. However, for $n>3$, solutions are limited to specific scenarios called homographic solutions. These solutions, also called self-similar solutions, remain constant under rotations and scaling, and are the only explicit solutions known for the $n$-body problem. A central configuration is a specific arrangement of the positions of the masses that gives rise to self-similar solutions.  Therefore, central configurations are counted modulo these types of motions.

The question of the existence and classification of central configurations is a historical problem that dates back to the eighteenth century. In 1767, Euler \cite{MR3297886} discovered three central configurations for $n=3$, where the masses lie on the same line for all time, called the collinear configurations\,. Five years later, Lagrange \cite{La} proved that an equilateral triangle, with each mass located at each vertex, is also a central configuration. Much later, in 1910, Moulton \cite{MR1503509} found that the exact count of collinear central configurations of $n$ bodies corresponds to $n!/2$. These results hold for any choice of the masses.

The task of determining the number of planar central configurations in the $n$-body problem for arbitrary positive masses is quite challenging. When $n=3$, there are five central configurations - two in an equilateral configuration and the rest in a collinear configuration. Long and Sun \cite{MR1892230} provided some partial and interesting results when considering 4 bodies with equal opposite masses $\beta>\alpha >0$\,. Additionally, Perez-Chavela and Santoprete \cite{MR2322818} demonstrated a unique convex non-collinear central configuration when the opposite masses are equal. In 2008, Albouy, Fu, and Sun \cite{MR2386651} proved that a convex central configuration is symmetric for one diagonal if and only if the masses of the two particles on the other diagonal are equal. The finiteness of central configurations is crucial for their counting, and Smale \cite{MR1754783} listed it as the 6th problem in 2002. It is important to note that the finiteness problem was settled by Albouy \cite{MR1320359} in 1995. When considering four equal masses, Hampton and Moeckel \cite{MR2207019} used symbolic calculations by computer to show that the number of central configurations is always between 32 and 8472, up to symmetry.

The polygonal central configurations, where each body is positioned at each vertex of a regular polygon, have been extensively analyzed. When all the masses are equal, then any regular $n$-gon can form a central configuration \cite{MR3469182}. Nested polygonal central configurations have also been studied\,. In 2003, Zhang and Zhou \cite{MR1963764} established some necessary and sufficient conditions for the nested polygonal solutions of planar $2n$-body problems\,. In 2012, Yu and Zhang  \cite{MR2946965} researched necessary and sufficient conditions for twisted nested central configurations formed by two twisted regular polygons with $n$ masses, respectively. The angle of twisting must be $\theta=0$ or $\theta=\pi/n$\,. Later, in 2015, they proved \cite{MR3297886} that under these conditions, the central configuration requires that the two polygons must have $n$ vertices.
In particular, Moczurad and Zgliczynski \cite{MR4017363} provided a computer-assisted classification of all central configurations with equal masses for 
$n=5,6,7$, using interval arithmetic methods. Their results include a broader set of configurations, and establish existence and symmetry properties, but are based on numerical and algorithmic approaches.

A rosette configuration is a coplanar arrangement where $n$ particles of mass $m_{1}$ are positioned at the vertices of a regular $n$-gon and $n$ particles of mass $m_{2}$ located at the vertices of another concentric $n$-gon but rotated by an angle of $\pi/n$ concerning the first. Additionally, there is an extra particle of mass $m_{0}$ at the center of mass. In 2006, Lei and Santoprete \cite{MR2247661} demonstrated that for $n \geq 3$ and every $\varepsilon>0$, with $\mu=m_0/m_1$ and $\varepsilon=m_2/m_1$, there exists a degenerate central configuration and a bifurcation.

For $n=5$, considering three bodies on the vertices of an equilateral triangle and two bodies on a perpendicular bisector, Llibre, and Mello \cite{Llibre:2008aa} showed the existence of three new families of planar central configurations\,. In 2013, Alvarez and Llibre \cite{MR3018443} characterized the planar central configurations of the $4$-body problem with masses $m_{1}=m_{2}\neq m_{3}=m_{4}$, which have an axis of symmetry\,. They showed that this $4$-body problem has exactly two classes of concave central configurations with the shape of a kite. Their proof was assisted by a computer. In 2021, Alvarez, Gassul, and Llibre \cite{MR4418916} classified the central configurations of the $5$-body problem, where the five bodies are positioned at the vertices of an equilateral pentagon with an axis of symmetry\,. They demonstrated that two unique classes of such equilateral pentagons provide central configurations: one concave equilateral pentagon and one convex equilateral pentagon, the regular one. On the other hand, in 2023, Deng and Hampton \cite{MR4502734} showed that the pentagonal configuration (not regular), with a cycle of five equal edges, has several results concerning central configurations satisfying this property\,. They also presented a computer-assisted proof of the finiteness of such configurations for any positive masses with a range of rational-exponent homogeneous potentials.

In the context of the $3$-body \cite{MR3297886} and $4$-body \cite{MR2386651} problems, it has been established that the only convex polygonal central configuration occurs when the bodies, each with equal mass, are arranged in a regular polygon. However, this is not true for the $6$-body problem \cite{MR1942919}, where configurations can consist of two nested triangles that are rotated relative to each other \cite{MR1963764}.

In this paper, we explore a specific subclass of planar central configurations for the Newtonian five-body problem with equal masses, which we call star central configurations. These are configurations where the five masses are positioned at equal angular separations around the center of mass but not necessarily at the same radial distance. Our main goal is to prove, through a fully analytic approach, that the only star central configuration for five equal masses corresponds to a regular pentagon.
It is important to emphasize that while Moczurad and Zgliczynski \cite{MR4017363} have previously classified all central configurations for five equal masses using computer-assisted interval arithmetic, their result includes a broader class of configurations and relies heavily on numerical methods. In contrast, we focus on a geometrically constrained subclass and provide a purely algebraic proof of uniqueness, offering complementary insight into the structure of central configurations. To our knowledge, no such analytical proof for this specific subclass has been presented in the existing literature.


The paper is organized as follows: in Section~\ref{preliminaries} and~\ref{basicresults}, we cover basic concepts and results related to central configurations and introduce some helpful ones for the rest of the paper\,. Following that, in Subsection~\ref{ss3.1}, we simplify the problem to a system equations with two variables, resulting in a new plausible system \eqref{sistema} within a significant domain $\hat{S}$ in the plane\,. Section~\ref{ss3.2} focuses on the main goal of the paper, which is the existence and uniqueness of star central configurations\,. To achieve this, we first establish the existence of these configurations in Theorem~\ref{th:existence}\,. Theorem~\ref{teorema} states the uniqueness for them, and the way of proving this involves dividing the domain $\hat{S}$ into 16 disjoint regions and demonstrating that the equations \eqref{sistema} are not satisfied in each region.
We would like to inform the reader that we only show explicitly the details of the calculations for the functions involved in the subregion $J_1$ defined in Subsection \ref{auxiliar}\,. This is because a similar strategy is followed for the other regions, and we want to avoid an overly extensive and tedious paper\,. The details for the rest of the subregions can be consulted oin the supplementary material.

\section{Preliminaries}  \label{preliminaries}
The $n$-body problem involves determining the motion of $ n $ point particles (without volume) in $ \mathbb{R}^d $ (where $ d = 1, 2, 3 $), each moving under the influence of Newton's law of gravitation. Assuming that the gravitational constant is $ G = 1 $, that each particle has a positive mass $ m_i $, and that the position of each particle is given by $ \mathbf{q}_i \in \mathbb{R}^d $, the equations of motion can be written as
\begin{equation}
m_i \ddot{\mathbf{q}}_i = -\sum_{j=1, \\ j \neq i}^n \frac{m_i m_j}{r_{ij}^3} (\mathbf{q}_i - \mathbf{q}_j) = \frac{\partial U}{\partial \mathbf{q}_i}, \quad i = 1, 2, \dots, n, \label{ecmov}
\end{equation}
where $ \dot{\ } $ denotes the derivative with respect to time, and $ r_{ij} = |\mathbf{q}_i - \mathbf{q}_j| $ is the Euclidean distance between particles $ i $ and $ j $. The function $ U: X \to \mathbb{R} $ is the Newtonian potential, given by
\begin{equation*}
U = \sum_{i < j} \frac{m_i m_j}{r_{ij}},
\end{equation*}
which represents the total potential energy of the system, where $ r_{ij} $ is the mutual distance between particles $ i $ and $ j $, and the position vector is $ \mathbf{q} = (\mathbf{q}_1, \dots, \mathbf{q}_n) \in \mathbb{R}^{nd} $.

We define the sets $ \triangle_{ij} = \{ \mathbf{q} \in \mathbb{R}^{nd} : \mathbf{q}_i = \mathbf{q}_j, i \neq j \} $ to represent all binary collisions in the system, and let the collision set $ \triangle $ be the union of all such sets, that is,
$
\triangle = \bigcup_{i \neq j} \triangle_{ij}.
$
Thus, the configuration space of the system, where no collisions occur, is given by
$
X = \mathbb{R}^{nd} \setminus \triangle.
$

Finally, we assume that the center of mass of the system is located at the origin. This condition is expressed as
$
m_1 \mathbf{q}_1 + \dots + m_n \mathbf{q}_n = \mathbf{0},
$
which is a first integral of the system, meaning that the total momentum of the system is conserved.

\begin{definition}\label{defII.1}
 A central configuration (CC) in the $n-$body problem for given masses and for some fixed time $t_0$, satisfies

\begin{equation}
-\ddot{\mathbf{q}}_k(t_0)=\lambda\mathbf{q}_k(t_0)\,, \quad k=1,...n,\label{cc}
\end{equation}

where $\lambda$ is a real constant.
\end{definition}    

\begin{definition}\label{ccpoligonalsemiregular}
Let $\mathbf{q}=(\mathbf{q}_1,...,\mathbf{q}_n)$ be a CC, it is named a Star Central Configuration (SCC), if each  $\mathbf{q}_i$, $i=1,...,n$, is placed at a vertex of a polygon for which the central angle $\theta$ (the angle between two consecutive ratios ) is $2\pi/n$.
\end{definition}
The SCC are easy to identify in polar coordinates because they can be written as

\begin{equation}
\mathbf{q}_i=r_i(\cos(\theta_i)\,,\sin(\theta_i))=(q_{i1},q_{i2})\,,  \quad \theta_i=\frac{2\pi (i-1)}{n}, \quad i=1,...n\,.\label{posicionesdeloscuerpos}
\end{equation}

It is well-known that if the configuration $ \mathbf{q}_0 $ is a central configuration (CC), then $ c \mathbf{q}_0 $ and $ A \mathbf{q}_0 $ are also CCs for any $ c \in \mathbb{R}^{+} $ and any $ A \in SO(3) $. In other words, the homothety (scaling) and rotations of a CC will also be CCs, and this provides a natural way to count them.  We can express the equations of motion \eqref{ecmov} in vector form as
$M \ddot{\mathbf{q}}_0 = \nabla U(\mathbf{q}_0),$
where $ M $ is the mass matrix, given by $ M = \text{diag}(m_1, \dots, m_n) $. By applying the matrix $ M $ to both sides of this equation and using the equations of motion, we obtain
$
\nabla U(\mathbf{q}_0) = \lambda \nabla I(\mathbf{q}_0),
$
which represents an optimization problem with $ \lambda $ as the Lagrange multiplier. Thus, a central configuration (CC) satisfies an optimization problem where the potential function $ U $ is minimized subject to the constraint that the moment of inertia $ I $ is constant. Specifically, the moment of inertia is given by
\begin{equation*}
I(\mathbf{q}) = \frac{1}{2} \sum_{j=1}^n m_j \mathbf{q}_j^2 = \frac{1}{4\tilde{m}} \sum_{i<j} m_i m_j r_{ij}^2,
\end{equation*}
where $ \tilde{m} = m_1 + \dots + m_n $ is the total mass and $ r_{ij} = |\mathbf{q}_i - \mathbf{q}_j| $ is the distance between particles $ i $ and $ j $.
Since both $ c \mathbf{q}_0 $ and $ A \mathbf{q}_0 $ are CCs, we can count the classes of CCs using this equivalence relation. To normalize the CCs, we set $ I = 1 $, which defines the sphere of masses $ S = \{ \mathbf{q} \in X : I = 1 \} $. That is, $ \mathbf{q}_0 $ is a normalized CC if and only if it is a critical point of the restriction of $ U(\mathbf{q}) $ to the set $ S $. It is important to note that the restriction of the potential $ U $ to $ S $ always attains its minimum at some $ \mathbf{q}_0 \in S $. The function $ I U^2 $ is homogeneous of degree zero and depends only on the mutual distances between particles. This property makes $ I U^2 $ invariant under rotations and scaling, which reflects how the shape of the configuration changes under such transformations. As a result, a critical point of $ I U^2 $ is a central configuration (CC), and vice versa. This function is referred to as the configuration measure.

\section{Basic results}\label{basicresults}

Let us consider the following set,
\begin{equation*}
\hat{S}=\left\{\mathbf{q}\in X\colon \mathbf{q}_1=(1,0)\right\},
\end{equation*}
which satisfies similar properties as the set $S$ mentioned in the previous Section \ref{preliminaries}.

\begin{proposition}
A configuration $\mathbf{q}\in S$ is a CC if and only if $|\mathbf{q}_1|^{-1} \bf{q}$ is a critical point of $U|_{\hat{S}}$.
\end{proposition}
\begin{proof}
We consider $\mathbf{q}\in S$ a CC\,. Let us define $\bar{\mathbf{q}}=R|\mathbf{q}_1|^{-1}\mathbf{q}$, with $\bar{\mathbf{q}}=(\bar{\mathbf{q_{1}}},...,\bar{\mathbf{q_{n}}})$ and $R$ a rotation matrix such that $\bar{\mathbf{q}}_1=(1,0)$ and then,

\begin{align*}
\nabla I(\mathbf{q})&=(m_1\mathbf{q}_1,...,m_n\mathbf{q}_n)=
\left(R^{-1}m_1|\mathbf{q}_1|\bar{\mathbf{q}_1},...,R^{-1}m_n|\mathbf{q}_1|\bar{\mathbf{q}_n}\right)
=R^{-1}|\mathbf{q}_1|\nabla I(\bar{\mathbf{q}})\,,\\
\nabla U(\mathbf{q})&=-\sum{m_im_j} \frac{(\mathbf{q}_i-\mathbf{q}_j)}{|\mathbf{q}_i-\mathbf{q}_j|^3}=-\sum{m_im_j} \frac{R^{-1}|{\mathbf{q}}_1|(\bar{\mathbf{q}}_i-\bar{\mathbf{q}}_j)}{|R^{-1}|^{3}|\mathbf{q}_1|^3|\bar{\mathbf{q}}_i-\bar{\mathbf{q}}_j|^3}=\frac{R^{-1}}{|\mathbf{q}_1|^2}\nabla U(\bar{\mathbf{q}})\,.
\end{align*}
Thus $\lambda \nabla I(\mathbf{q})=\lambda |\mathbf{q}_1|\nabla I(\bar{\mathbf{q}})=|\mathbf{q}_1|^{-2}\nabla U(\bar{\mathbf{q}})=\nabla U(\mathbf{q})$\,. 
This allows us to write this equation as $\bar{\lambda}\nabla I(\bar{\mathbf{q}})=\nabla U(\bar{\mathbf{q}})$ where we name $\bar{\lambda}=\lambda|\mathbf{q}_1|^3$\,. Therefore  $\bar{\mathbf{q}}$ is a critical point of $U$ restricted to $\hat{S}$.
\end{proof}

\begin{proposition}
The restriction of the potential $U$ to $\hat{S}$ always attains its minimum at some $\mathbf{q}\in \hat{S}$
\end{proposition}
\begin{proof}

We know that $ U $ restricted to $ S $ attains a minimum at $ \mathbf{q}^* $. We then define $ \bar{\mathbf{q}}^* = R | \mathbf{q}_1^* |^{-1} \mathbf{q}^* $, which represents a rescaling and rotation of $ \mathbf{q}^* $. In this way, $ \mathbf{q}^* $ is a minimum of $ I U^2 $, since $ I U^2 $ is a zero-degree homogeneous function. Therefore, $ \bar{\mathbf{q}}^* $ is a minimum of $ U $ restricted to the set $ \hat{S} $.
\end{proof}		

\section{$5$-body problem}\label{ss3.1}
We are considering five bodies of equal mass forming a SCC in $ \hat{S} $. By following the equations in \eqref{cc}, we obtain the system

\begin{equation}
\lambda \mathbf{q}_i = m \sum_{i \neq j}^5 \frac{\mathbf{q}_i - \mathbf{q}_j}{r_{ij}^3}. \label{7}
\end{equation}

Writing the positions in polar coordinates as specified in \eqref{posicionesdeloscuerpos}, with $ r_1 = 1 $ and the center of mass fixed at the origin, simplifies the system \eqref{7} to the following equations:

\begin{equation}
\begin{aligned}
\lambda_{ik}(r_3, r_5) &= \frac{m}{q_{ik}(r_3, r_5)} \sum_{i \neq j} \frac{q_{ik}(r_3, r_5) - q_{jk}(r_3, r_5)}{r_{ij}^3}, \quad k = 1, 2. \quad q_{ik} \neq 0\,.\\
\lambda_{12}(r_3, r_5) &= 0, 
\end{aligned} \label{sistema}
\end{equation}

We express the parameter $ \lambda $ as a function that depends on $ r_3 $ and $ r_5 $. According to Definition \ref{defII.1}, we have $ \lambda_{ik} = \lambda_{lm} $ for all $ i, k, l, m = 1, 2 $. By equating all the equations in \eqref{sistema}, except for $ \lambda_{12}(r_3, r_5) $, we obtain the solution we are seeking. The domain $\hat{S}$ for the system \eqref{sistema} is described by the set 
\begin{equation*}
\hat{S}=\left\{(r_3,r_5)\in\mathbb{R}^2 \colon r_3>0, r_5>0, r_5>r_3-b/2,r_5>\dfrac{ar_3-a}{2}\right\}\,,
\end{equation*}
where $a=\sqrt{5}+1$ and $b=\sqrt{5}-1$, (see Fig.~\ref{dominio_fig}).
\begin{figure}[ht]
	\centering
		\includegraphics[scale=0.8]{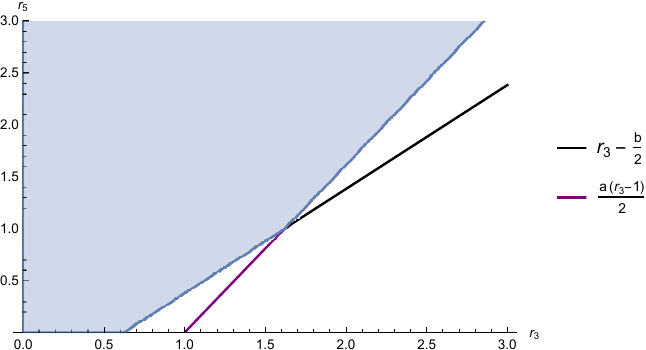}
		\caption{domain $\hat{S}$}
		\label{dominio_fig}
\end{figure}

\section{Uniqueness}
\label{ss3.2}

The main objective of this paper is to demonstrate that a solution to the SCC within the domain $ \hat{S} $, which satisfies equation \eqref{sistema}, corresponds exclusively to a regular pentagon and not to any other configuration. This finding leads to two important conclusions: first, it confirms that the regular pentagon is indeed a valid solution, and second, it establishes the uniqueness of this solution.

\begin{theorem}\label{th:existence}
The point $(1,1)$ is a minimum for $IU^{2}$.
\end{theorem}

\begin{proof}
A solution for \eqref{sistema} is $(r_3,r_5)=(1,1)$, thus is a critical point of the configuration measure $IU^{2}$\,. Calculating the Hessian at $(1,1)$ gives us 
\[ H(1,1) =
\left( \begin{array}{cc}
\frac{5}{4}(25 + 13 \sqrt{5}) & -\frac{5}{8}(25 + 13 \sqrt{5})\\[2mm]
-\frac{5}{8}(25 + 13 \sqrt{5}) & \frac{5}{4}(5 + 7 \sqrt{5})
\end{array} \right)\,. \]
The determinant of $H(1,1)$ is equal to $125(85 + 31 \sqrt{5})/32>0$, and the principal minors of $H(1,1)$ are greater than zero\,. Therefore, $IU^2$ is a convex function in a neighborhood of $(1,1)$, which is a minimal point for $IU^2$.
\end{proof}

\begin{theorem}\label{teorema}
If $(r_{3},r_{5})\neq(1,1)$, there is no solution for the system \eqref{sistema}.
\end{theorem}
To prove Theorem \ref{teorema}, we divide the domain $\hat{S}$ in 16 regions, listed as follows (see Fig.~\ref{fig2})\,,

\begin{equation*}
\begin{split}
J_1&=\left\{(r_3,r_5)\in \mathbb{R}^2 \colon 0<r_3\leq \dfrac{b}{2}, 0<r_5\leq \frac{b}{2} \right\}\\
J_2&=\left\{(r_3,r_5)\in \mathbb{R}^2 \colon 0<r_3\leq \frac{b}{2}, b/2<r_5\leq 1 \right\}\\
J_3&=\left\{(r_3,r_5)\in \mathbb{R}^2 \colon b/2<r_3<1, \frac{2-b}{2}\leq r_5< \frac{b}{2} \right\} \\
J_4&=\left\{(r_3,r_5)\in \mathbb{R}^2 \colon \frac{b}{2}<r_3<r_5+\frac{b}{2}, 0\leq r_5< \frac{2-b}{2} \right\}\\
J_5&=\left\{(r_3,r_5)\in \mathbb{R}^2 \colon 1<r_3<r_5+\frac{b}{2}, \frac{2-b}{2}\leq r_5< \frac{b}{2} \right\}\\
J_6&=\left\{(r_3,r_5)\in \mathbb{R}^2 \colon b<r_3<r_5+\frac{b}{2}, r_3-\frac{b}{2}\leq r_5< 1 \right\} \\
J_7&=\left\{(r_3,r_5)\in \mathbb{R}^2 \colon \frac{b}{2}<r_3<1, \frac{b}{2}\leq r_5<1 \right\}\\
J_8&=\left\{(r_3,r_5)\in \mathbb{R}^2 \colon 1<r_3<b, \frac{b}{2}\leq r_5<1 \right\}\\
J_9&=\left\{(r_3,r_5)\in \mathbb{R}^2 \colon 0<r_3<1, 1\leq r_5<\infty \right\}\\
J_{10}&=\left\{(r_3,r_5)\in \mathbb{R}^2 \colon 1<r_3<\dfrac{2}{b}, 1+b\leq r_5<\infty \right\}
\end{split}
\end{equation*}
\begin{equation*}
\begin{split}
J_{11}&=\left\{(r_3,r_5)\in \mathbb{R}^2 \colon \frac{b}{2}<r_3<\dfrac{2}{a}r_{5}+1, 1\leq r_5<3.036	 \right\}
\\
J_{12}&=\left\{(r_3,r_5)\in \mathbb{R}^2 \colon 1.3<r_3<\dfrac{2}{b}, 1\leq r_5<2.05	 \right\}\\
J_{13}&=\left\{(r_3,r_5)\in \mathbb{R}^2 \colon 1<r_3<1.3, 1.4\leq r_5<2.05	 \right\}\\
J_{14}&=\left\{(r_3,r_5)\in \mathbb{R}^2 \colon 1<r_3<2/b, 2.05\leq r_5<1+b	 \right\}\\
J_{15}&=\left\{(r_3,r_5)\in \mathbb{R}^2 \colon 2/b<r_3<\dfrac{2}{a}r_{5}+1, 3.036\leq r_5<\infty	 \right\}\\
J_{16}&=\left\{(r_3,r_5)\in \mathbb{R}^2 \colon 1<r_3<1.3, 1<r_5<1.4 	 \right\}
\end{split}
\end{equation*}

\begin{figure}[ht]
	\centering
		\includegraphics[scale=0.48]{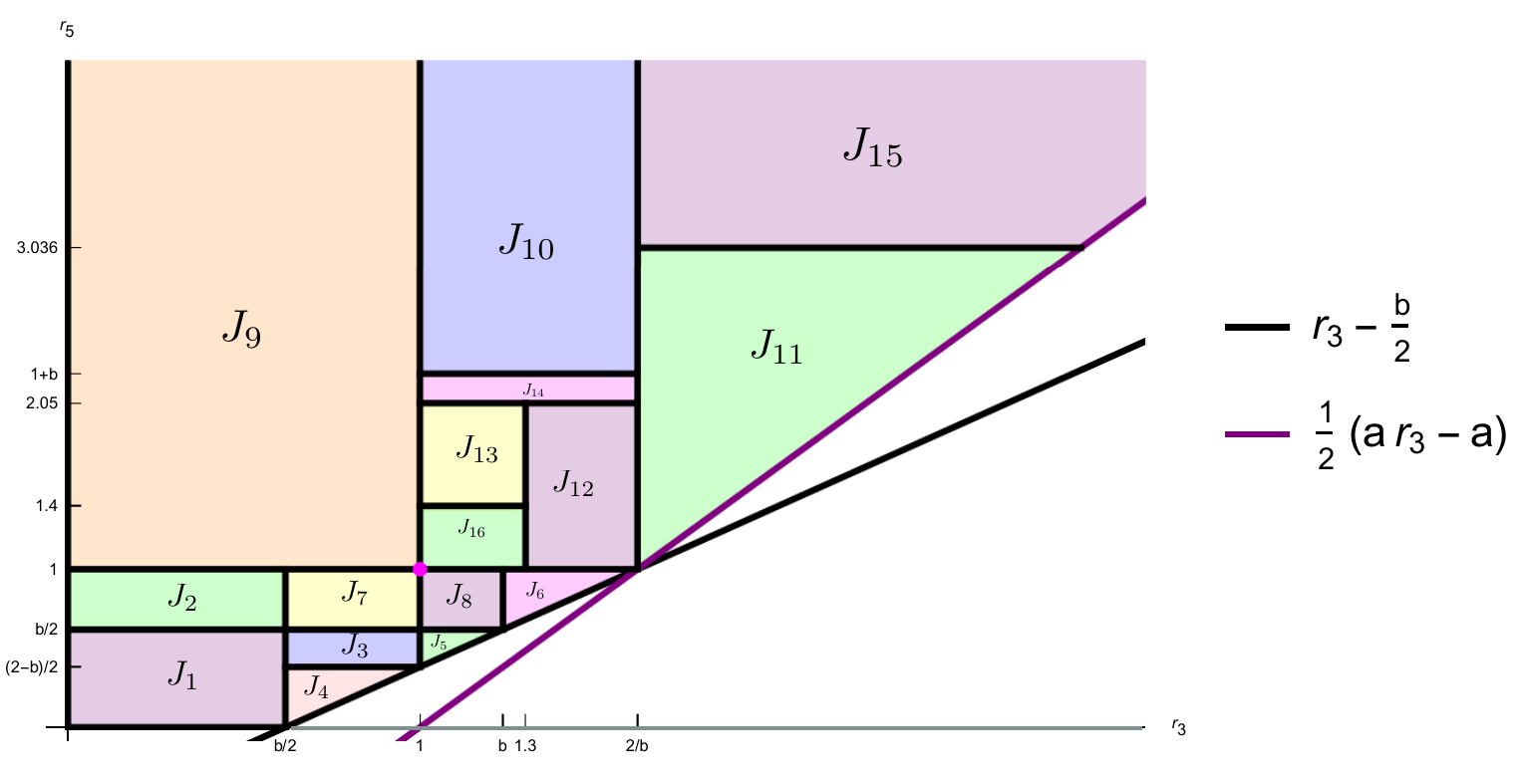}
		\caption{regions $J_{1},...,J_{16}$.}
		\label{fig2}
\end{figure}

We assert that there is no region $J_i$ where the Central Configuration equations are satisfied\,. Specifically, for each $J_n$ where $n=1,...,16$, there exist a pair of equations $\lambda_{ij}(r_3,r_5)$ and $\lambda_{kl}(r_3,r_5)$, such that $\lambda_{ij}(r_3,r_5)\neq\lambda_{kl}(r_3,r_5)$ for all $r_3$ and $r_5$ in $J_n$\,. Because of this, we will begin by proving fundamental properties to achieve the desired result.

\subsection{Properties of auxiliary functions of $J_1$}\label{auxiliar}

\begin{definition}
A family of functions $f_{\alpha}(r)$ with $\alpha\in I$ is strictly increasing with respect to the parameter $\alpha$ if for $\alpha_{1}<\alpha_{2}$, it satisfies $f_{\alpha _{1}}(r)<f_{\alpha _{2}}(r)$ for all $r$ in the appropriate domain.

A family of functions $f_{\alpha}(r)$ with $\alpha\in I$ is strictly decreasing with respect to the parameter $\alpha$ if for $\alpha_{1}<\alpha_{2}$, it satisfies $f_{\alpha _{1}}(r)>f_{\alpha _{2}}(r)$ for all $r$ in the appropriate domain.
\end{definition}

For the region $J_{1}$, we define the following family of functions 
\begin{equation*}
\lambda_{ik_{\eta}}(r_{5})\colonequals\lambda_{ik}\left(\frac{b}{2+\eta},r_{5}\right)\,,  \quad \eta\in(0,\infty)\,.
\end{equation*}
 
We are currently analyzing the families of functions $\lambda_{11_{\eta}}(r_{5})$ and $\lambda_{31_{\eta}}(r_{5})$, which can be expressed as $\lambda_{11_{\eta}}(r_{5})=e_{1_{\eta}}(r_{5})+e_{2_{\eta}}(r_{5})+e_{3_{\eta}}(r_{5})$ and $\lambda_{31_{\eta}}(r_{5})=f_{1_{\eta}}(r_{5})+f_{2_{\eta}}(r_{5})+f_{3_{\eta}}(r_{5})+f_{4_{\eta}}(r_{5})$\,. For into $J_{1}$, the following properties are needed to prove Propositions \ref{aux2} to \ref{auxextra1} and can be verified through straightforward computation\,. 

\begin{properties}\label{es}{$ $}
\begin{itemize}
\item The family of function $e_{1_{\eta}}(r_{5})$ is strictly increasing\,. 
\item The family of function $e_{2_{\eta}}(r_{5})$ is strictly decreasing\,. 
\item The family of function $e_{3_{\eta}}(r_{5})$ is strictly decreasing for $\eta\in(0.52,\infty)$.
\item The family of function $e_{3_{\eta}}(r_{5})$ is strictly decreasing for $\eta\in(0,0.52] $ and $r_{5}\in\left(0.12874,b/2\right]$.
\item Let $r_{5_{\eta}}:=\displaystyle \max_{r_{5}\in(0,b/2]}e_{3_{\eta}}$, thus $r_{5_{\eta}}<\infty$ and $e_{3_{\eta}}(r_{5_{\eta}})\leq 1.0696$.
\end{itemize}
\end{properties}

\begin{properties}\label{fs}{$ $}
\begin{itemize}
Below are the properties of the functions with respect to the variable $r_5$ and parameter $\eta$:

\item The family of functions $f_{1_{\eta}}(r_{5})$ and $f_{2_{\eta}}(r_{5})$ are strictly increasing.
\item The family of function $f_{3_{\eta}}(r_{5})$ is strictly decreasing.
\item The family of function $f_{4_{\eta}}(r_{5})$ is strictly decreasing for $\eta\in(1.16249,7.60839]$, and it is strictly increasing for $\eta\in(7.60839,\infty)$.
\item The functions of the family $f_{1_{\eta}}(r_{5})$ do not depend on $r_{5}$.
\item The function $\lambda_{31_{0}}(r_{5})$ is monotonically decreasing.

For $\eta\in[0,0.02]$, the following statements hold:
\item The functions of the family $f_{2_{\eta}}(r_{5})$ monotonically decrease as $r_{5}\in(0,b/2]$.
\item The functions of the family $f_{3_{\eta}}(r_{5})$ monotonically decrease as $r_{5}\in(0,r_{5}^{*}(\eta)]$.
\item The functions of the family $f_{3_{\eta}}(r_{5})$ monotonically increase as $r_{5}\in(r_{5}^{*}(\eta),b/2]$.
\item The functions of the family $f_{4_{\eta}}(r_{5})$ monotonically increase as $r_{5}\in(0,\hat{r_{5}}(\eta)]$.
\item The functions of the family $f_{4_{\eta}}(r_{5})$ monotonically decrease as $r_{5}\in(\hat{r_{5}}(\eta),b/2]$.

The functions $r_{5}^{*}(\eta)$ and $\hat{r_{5}}(\eta)$ are monotonically decreasing, and it holds that $r_{5}^{*}(0)=0.417957$, $r_{5}^{*}(0.02)=0.397798$, and $\hat{r_{5}}(0)=0.156497$, $\hat{r_{5}}(0.02)=0.148828$.
\end{itemize}
\end{properties}

\begin{properties}\label{fsp}
For $\eta\in[0,0.02]$ the following is true:
\begin{itemize}
\item The functions of the families $f_{2_{\eta}}'(r_{5})$ and $f_{3_{\eta}}'(r_{5})$ are monotonically increasing.
\item The functions of the family $f_{4_{\eta}}'(r_{5})$ are monotonically decreasing for $r_{5}\in(0,\hat{r_{5}}(\eta)]$.
\item The family of functions $f_{2_{\eta}}'(r_{5})$ is strictly decreasing.
\item The family of functions $f_{3_{\eta}}'(r_{5})$ is strictly increasing.
\end{itemize}
\end{properties}

\begin{properties}\label{dfs}
Let $\left(d_{\eta}f_{i_\eta}\right)(r_5):=\frac{d}{d\eta}f_{i_\eta}(r_5)$, $i=1,...,4$.
\begin{itemize}
\item The families of functions $\left(d_{\eta}f_{1_\eta}\right)(r_5)$ and $\left(d_{\eta}f_{2_\eta}\right)(r_5)$ are strictly increasing.
\item The family of functions $\left(d_{\eta}f_{3_\eta}\right)(r_5)$ is strictly increasing for $\eta\in(0,2.1722)$.
\item The family of functions $\left(d_{\eta}f_{3_\eta}\right)(r_5)$ is strictly decreasing for $\eta\in[3.76584,\infty)$.
\item The family of functions $\left(d_{\eta}f_{4_\eta}\right)(r_5)$ is strictly increasing for $\eta\in(4.52952,7.60839]$.
\item The functions of the family $\left(d_{\eta}f_{2_\eta}\right)(r_5)$ are monotonically decreasing.
\item The functions of the family $\left(d_{\eta}f_{3_\eta}\right)(r_5)$ are monotonically increasing\,. 
\item The functions of the family $\left(d_{\eta}f_{4_\eta}\right)(r_5)$ are negative for $\eta\in(1.16249,7.60839]$.
\item The functions of the family $\left(d_{\eta}f_{4_\eta}\right)(r_5)$ are monotonically increasing for the following intervals: 
\begin{itemize}
\item $r_{5}\in(0,b/2]$ and $\eta\in(1.64631,4.52952]$.
\item $r_{5}\in(0.1,b/2]$ and $\eta\in(1.16249,1.64631]$.
\item $r_{5}\in(0.45,b/2]$ and $\eta\in(0.11,1.16249]$.
\item $r_{5}\in(0.582276,b/2]$ and $\eta\in(0.02,0.11]$.
\item $r_{5}\in(0.0785,0.1]$ and $\eta\in(1.52721,1.64631]$.
\end{itemize}
\item The functions of the family $\left(d_{\eta}f_{4_\eta}\right)(r_5)$ are concave for $\eta\in(2.09792,6.24402)$.
\item The functions of the family $\left(d_{\eta}f_{4_\eta}\right)(r_5)$ are convex for the following intervals: 
\begin{itemize}
\item $r_{5}\in(0,0.582276]$ and $\eta\in(0.02,0.11]$.
\item $r_{5}\in(0,0.1]$ and $\eta\in(1.16249,1.52721]$.
\item $r_{5}\in(0,0.756]$ and $\eta\in(1.52721,1.64631]$.
\end{itemize}
\end{itemize}
\end{properties}

\begin{properties}\label{pfs} 
Let $\left(d_{\eta}f_{i_{\eta}}\right)'(r_{5}):=\frac{d}{dr_{5}}\left(\left(d_{\eta}f_{i_{\eta}}\right)(r_{5})\right)$, $i=1,...,4$.
\begin{itemize}
\item The functions of the family $(d_{\eta}f_{2_{\eta}})'(r_{5})$  are concave.
\item The functions of the family $(d_{\eta}f_{2_{\eta}})'(r_{5})$ are monotonically increasing.
\item The functions of the family $\left(d_{\eta}f_{3_{\eta}}\right)'(r_{5})$ are concave for $r_{5}\in(0,0.1]$ and $\eta\in(0,0.037)$.
\item The functions of the family $\left(d_{\eta}f_{3_{\eta}}\right)'(r_{5})$ are monotonically decreasing for $\eta\in[0.037,\infty)$.
\item The functions of the family $\left(d_{\eta}f_{3_{\eta}}\right)'(r_{5})$ are monotonically decreasing for $r_{5}\in(0.1,b/2]$ and $\eta\in(0,0.037)$.
\end{itemize}
\end{properties}

\begin{properties}\label{spfs}
Let $(d_{\eta}f_{2_{\eta}})''(r_{5})$ and $(d_{\eta}f_{3_{\eta}})''(r_{5})$ be the families for which they are the second derivative respect to $r_{5}$ of the functions of $\left(d_{\eta}f_{2_{\eta}}\right)(r_{5})$ and $\left(d_{\eta}f_{3_{\eta}}\right)(r_{5})$, respectively.

\begin{itemize}
\item The family of functions $(d_{\eta}f_{2_{\eta}})''(r_{5})$ is strictly increasing.
\item The functions of the family $(d_{\eta}f_{2_{\eta}})''(r_{5})$ are monotonically decreasing and convex.
\item The functions of the family $(d_{\eta}f_{3_{\eta}})''(r_{5})$ are concave for $\eta\in(0.037,\infty)$.
\item The functions of the family $(d_{\eta}f_{3_{\eta}})''(r_{5})$ are concave for $r_{5}\in(0.4,b/2]$ and $\eta\in(0,0.037]$.
\item The functions of the family $(d_{\eta}f_{3_{\eta}})''(r_{5})$ are monotonically increasing for $r_{5}\in(0.2,b/2]$ and $\eta\in(0,0.037]$.
\item The functions of the family $(d_{\eta}f_{3_{\eta}})''(r_{5})$ are monotonically increasing for $r_{5}\in(0,0.3]$ and $\eta\in(0.26,\infty)$.
\end{itemize}
\end{properties}

Although the domain of the system excludes $r_{3}=0$ and $r_{5}=0$, this system is well-defined. We are going to use these values at the border of the domain to justify the main results.

\begin{proposition}\label{aux2}
For into $J_{1}$, the family of functions $f_{1_{\eta}}(r_{5})+f_{2_{\eta}}(r_{5})+f_{3_{\eta}}(r_{5})$ is strictly increasing.
\end{proposition}

\begin{proof}
We can establish the strict positivity of the function $\left(d_{\eta}f_{1_{0}}\right)(r_{5})+\left(d_{\eta}f_{2_{0}}\right)(r_{5}) +\left(d_{\eta}f_{3_{0}}\right)(r_{5})$ by comparing it with the functions $\left(d_{\eta}f_{1_{0}}\right)(r_{5})+\left(d_{\eta}f_{2_{0}}\right)(r_{5})$ and $-\left(d_{\eta}f_{3_{0}}\right)(r_{5})$\,. We know these functions monotonically decrease due to Properties~\ref{dfs}\,. Now, let us consider the following piecewise functions:

\begin{equation*}
L_{3}(r_5)= \left\{ \begin{array}{lll}
             -\left(d_{\eta}f_{3_{0}}\right)(0)=3.413203,  &0<r_{5}\leq 0.2, \\
             -\left(d_{\eta}f_{3_{0}}\right)(0.2)=1.90132, & 0.2<r_{5}\leq 0.45,\\
             -\left(d_{\eta}f_{3_{0}}\right)(0.45)=0.83961, & 0.45<r_{5}\leq 1,
\end{array}
   \right.
\end{equation*}

\begin{equation*}
L_{4}(r_5)= \left\{ \begin{array}{lll}
             \left(d_{\eta}f_{1_{0}}\right)(0.2)+ \left(d_{\eta}f_{2_{0}}\right)(0.2)=3.51384,  &0<r_{5}\leq 0.2, \\
             \left(d_{\eta}f_{1_{0}}\right)(0.45)+ \left(d_{\eta}f_{2_{0}}\right)(0.45)=1.92656, & 0.2<r_{5}\leq 0.45,\\
             \left(d_{\eta}f_{1_{0}}\right)(1)+ \left(d_{\eta}f_{2_{0}}\right)(1)=1.05376, & 0.45<r_{5}\leq 1.
\end{array}
   \right.
\end{equation*}
In this way $-\left(d_{\eta}f_{3_{0}}\right)(r_{5})<L_{3}(r_{5})<L_{4}(r_{5})<\left(d_{\eta}f_{1_{0}}\right)(r_{5})+\left(d_{\eta}f_{2_{0}}\right)(r_{5})$\,. Therefore, we conclude that $\left(d_{\eta}f_{1_{\eta}}\right)(r_{5})+\left(d_{\eta}f_{2_{\eta}}\right)(r_{5})+\left(d_{\eta}f_{3_{\eta}}\right)(r_{5})>0$.
\end{proof}

\begin{proposition}\label{otraaux}
For into $J_{1}$, the family of functions $(d_{\eta}f_{1_{\eta}})(r_{5})+(d_{\eta}f_{2_{\eta}})(r_{5})+(d_{\eta}f_{3_{\eta}})(r_{5})$ is strictly increasing.
\end{proposition}
\begin{proof}
Since Properties \ref{dfs}, the statement is true for $\eta\in(0,2.1722)$\,. We divide the domain of $\eta$ into $[2.1722, 3.76584]\cup (3.76584,\infty)$.

\underline{Part I}\,. For $\eta\in[2.1722, 3.76584]$\,. We divide the interval of $r_{5}$ into $(0,0.3] \cup (0.3,b/2]$\,. We will analyze and compare the functions $\left(d_{\eta}f_{1_{\eta}}\right)(r_{5}) + \left(d_{\eta}f_{2_{\eta}}\right)(r_{5})$ and $-\left(d_{\eta}f_{3_{\eta}}\right)(r_{5})$ within this interval.

\begin{itemize}
\item Part Ia: for \(r_{5} \in (0, 0.3]\)\,. Based on Properties \ref{dfs}, we can deduce that the minimum value of the family $\left(d_{\eta} f_{1_{\eta}}\right)(r_{5}) + \left(d_{\eta} f_{2_{\eta}}\right)(r_{5})$ is $\left(d_{\eta} f_{1_{2.1722}}\right)(0.3) + \left(d_{\eta} f_{2_{2.1722}}\right)(0.3) = 4.12926$\,. Furthermore, as the functions of the family $-\left(d_{\eta} f_{3_{\eta}}\right)(r_{5})$ are monotonically decreasing, we can treat the value of $-\left(d_{\eta} f_{3_{\eta}}\right)(0)    $ as a function of $\eta$\,. An upper bound of this function is $-\left(d_{\eta} f_{3_{2.1722}}\right)(0)=0.319658$\,. Therefore, $\left(d_{\eta} f_{1_{\eta}}\right)(r_{5}) + \left(d_{\eta} f_{2_{\eta}}\right)(r_{5})>- \left(d_{\eta} f_{3_{\eta}}\right)(r_{5})$.

\item Part Ib: for $r_{5}\in (0.3,b/2]$\,. Utilizing Properties \ref{dfs}, we can establish a lower bound for the family $\left(d_{\eta}f_{1_{\eta}}\right)(r_{5}) + \left(d_{\eta}f_{2_{\eta}}\right)(r_{5})$, which equates to $\left( d_{\eta}f_{1_{2.1722}}\right)(b/2) + \left(d_{\eta}f_{2_{2.1722}}\right)(b/2)=1.78421$\,. Subsequently, let us analyze $-\left(d_{\eta}f_{3_{\eta}}\right)(0.3)$, a function depending on $\eta$, achieving its maximum value at $\eta=2.1722$\,. Hence, $-\left(d_{\eta}f_{3_{2.1722}}\right)(0.3)=0.20796$\,. Thus, we can infer that $\left(d_{\eta} f_{1_{\eta}}\right)(r_{5}) + \left(d_{\eta} f_{2_{\eta}}\right)(r_{5}) > -\left(d_{\eta} f_{3_{\eta}}\right)(r_{5})$.
\end{itemize}

\underline{Part II}: for $\eta\in(3.76584,\infty)$\,. According to the Properties \ref{dfs}, we can observe that the functions of the families $\left(d_{\eta}f_{1_{\eta}}\right)(r_{5})+\left(d_{\eta}f_{2_{\eta}}\right)(r_{5})$ and $-\left(d_{\eta}f_{3_{\eta}}\right)(r_{5})$ are monotonically decreasing\,. This means that the minimum value of $\left(d_{\eta}f_{1_{\eta}}\right)(r_{5})+\left(d_{\eta}f_{2_{\eta}}\right)(r_{5})$ is $\displaystyle \lim_{\eta\rightarrow \infty} \left[\left(d_{\eta}f_{1_{\eta}}\right)(b/2 )+\left(d_{\eta}f_{2_{\eta}}\right)(b/2)\right]=1.80902$\,. Likewise, the maximum value of the family $-\left(d_{\eta}f_{3_{\eta}}\right)(r_{5})$ can be calculated as $-\left(d_{\eta}f_{3_{\eta}}\right)(0)=0.2794$\,. Therefore, we can conclude the inequality $\left(d_{\eta} f_{1_{\eta}}\right)(r_{5}) + \left(d_{\eta} f_{2_{\eta}}\right)(r_{5})>-\left(d_{\eta} f_{3_{\eta}}\right)(r_{5})$.
\end{proof}

\begin{proposition}\label{aux3}
For into $J_{1}$, the functions of the family $\left(d_{\eta}f_{1_{\eta}}\right)(r_{5})+\left(d_{\eta}f_{2_{\eta}}\right)(r_{5})+\left(d_{\eta}f_{3_{\eta}}\right)(r_{5})$ are monotonically decreasing\,. 
\end{proposition}

\begin{proof}
According to Properties \ref{pfs}, the objective is to demonstrate the inequality $-(d_{\eta}f_{2_{\eta}})'(r_{5})>(d_{\eta}f_{3_{\eta}})'(r_{5})$\,. We plan to partition the interval $\eta$ into two subsets: $(0,0.037)\cup[0.037,\infty)$.

\underline{Part I}: for $\eta\in(0,0.037)$\,. By considering concavity of the functions of the family $(d_{\eta}f_{3_{\eta}})'(r_{5})$ (Properties \ref{pfs}), we will divide the interval of $r_ {5}$ into $(0,0.1]\cup(0.1,b/2]$.
\begin{itemize}
\item Part Ia: for $r_{5}\in(0,0.1]$\,. Since the Properties \ref{pfs}, the functions of the family $\left(d_{\eta}f_{3_{\eta}}\right)'(r_{5})$ are concave, we examine the tangent lines at the points $(0,(d_{\eta}f_{3_{\eta}})'(0))$ and $(0.1,(d_{\eta}f_{3_{\eta}})'(0.1))$, denoting the respective straight lines as $L_{150_{\eta}}(r_{5})$ and $L_{151_{\eta}}(r_{5})$\,. The families $L_{150_{\eta}}(r_{5})$ and $L_{151_{\eta}}(r_{5})$ intersect at $r_{5}^{*}(\eta)$, where $r_{5}^{*}=\{r_{5}|L_{150_{\eta}}(r_{5})=L_{151_{\eta}}(r_{5})\}$, and $L_{150_{\eta}}(r_{5}^{*})$ is strictly increasing, reaching its maximum value at $L_{150_{0.37}}(r_{5}^{*}(0.37)) = 8.72038$\,. Furthermore, since the functions of the family $-\left(d_{\eta}f_{2_{\eta}}\right)'(r_{5})$ are monotonically decreasing (Properties \ref{pfs}), a straightforward calculation shows that $-\left(d_{\eta}f_{2_{\eta}}\right)(0.1)$ corresponds to a strictly increasing function, with $\min_{\eta\in (0,0.037)}\left(d_{\eta} f_{2_{\eta}}\right)(0.1)=\left(d_{\eta}f_{2_{0}}\right)(0.1)=17.0956$\,. (See Fig.~\ref{fig3a}).

Hence, $-(d_{\eta}f_{2_{\eta}})'(r_{5})>(d_{\eta}f_{3_{\eta}})'(r_{5})$.

\item Part Ib: for $r_{5}\in(0.1,b/2]$\,. For a better analysis of the functions, we will divide the interval of
$r_{5}$ into $(0.1, 0.3] \cup (0.3, 0.5] \cup (0.5, b/2]$\,. Due to the convexity of the functions in the family $-\left(d_{\eta}f_{2_{\eta}}\right)'(r_{5})$ (Properties \ref{pfs}), we can define the family of tangent line functions $L_{152_{\eta}}(r_{5}), L_{153_{\eta}}(r_{5}), L_{154_{\eta}}(r_{5})$\,. These tangent lines are defined at the upper limits of each subinterval of $r_{5}$ (See Fig.~\ref{fig3b})\,. Furthermore, it is important to note that the family of functions $\left(d_{\eta}f_{3_{\eta}}\right)'(r_{5})$ exhibits monotonic decreasing behavior (Proposition \ref{pfs}).

\begin{figure}[ht]
 \centering
  \subfloat[Graphs of the functions with labels for $\eta=0.01405$.]{
   \label{fig3a}
   \includegraphics[width=0.43\textwidth]{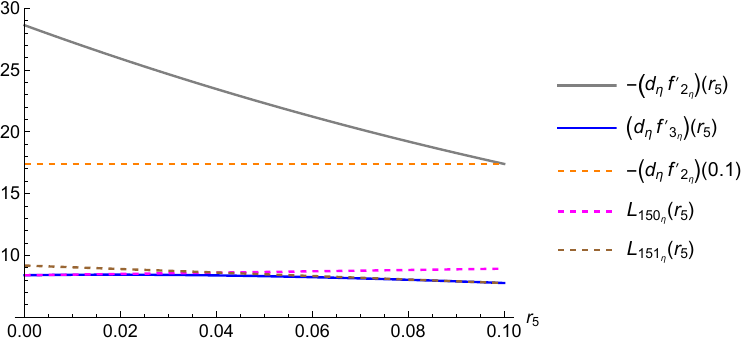}} 
 \subfloat[Graphs of the functions with labels for $\eta=0.0273$.]{
   \label{fig3b}
   \includegraphics[width=0.43\textwidth]{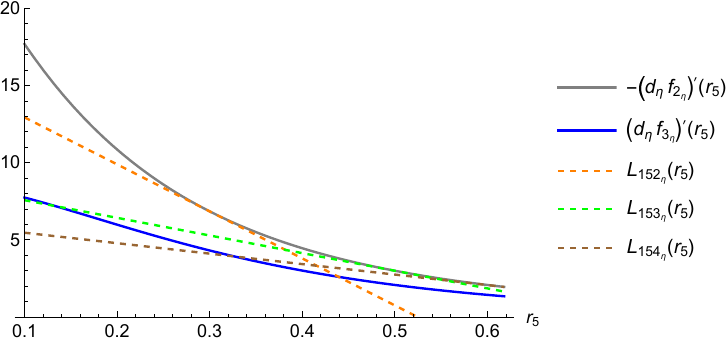}}
   \caption{ }
\end{figure}

\begin{itemize}
\item Subinterval $r_{5}\in(0.1,0.3]$\,. We first divide this interval into $(0.1,0.25]\cup(0.25,0.3]$.
\begin{itemize}
\item For $r_{5}\in (0.1,0.25]$\,. We compare the families $(d_{\eta}f_{3_{\eta}})'(0.1)$ and $L_{152_{\eta}}(0.25)$, with their respective upper and lower values $(d_{\eta}f_{3_{0.00839}})'(0.1)=7.75675$ and $L_{152_{0}}(0.25)=8.16295$, respectively\,. Thus, we have $(d_{\eta}f_{3_{\eta}})'(0.1)<L_{152_{\eta}}(0.25)$.
\item For $r_{5}\in(0.25,0.3]$\,. Using a similar approach as before, we compare the upper and lower bounds of $(d_{\eta}f_{3_{\eta}})'(0.25)$ and $L_{152_{\eta}}(0.3)$\,. The bounds are given by $(d_{\eta}f_{3_{0}})'(0.25)=5.32066$ and $L_{152_{0}}(0.3)=6.69021$, respectively\,. Thus, it follows that $(d_{\eta}f_{3_{\eta}})'(0.25)<L_{152_{\eta}}(0.3)$.
\end{itemize}
Therefore, we have that $(d_{\eta}f_{3_{\eta}})'(r_{5})< L_{152_{\eta}}(r_{5})$.

\item Subinterval $r_{5}\in(0.3,0.5]$\,. As in the previous cases, we will compare the upper and lower bounds for the families  $(d_{\eta}f_{3_{\eta}})'(r_{5})$ and $ L_{153_{\eta}}(r_{5})$ within the following subintervals of $r_{5}$ as follows $(0.3,0.35]\cup(0.35,0.4]\cup(0.4,0.45]\cup(0.45,0.5]$\,. 
\begin{itemize}
\item For $r_{5}\in(0.3,0.35]$\,. The upper and lower bounds are $(d_{\eta}f_{3_{\eta}})'(0.3)=4.52419$ and $L_{153_{0}}(0.35)=4.62889$, respectively\,. Thus, $(d_{\eta}f_{3_{\eta}})'(0.3)<L_{153_{\eta}}(0.35)$.
\item For $r_{5}\in(0.35,0.4]$\,. The upper and lower bounds are $(d_{\eta}f_{3_{\eta}})'(0.35)=3.81145$ and $L_{153_{0}}(0.4)=4.07061$, respectively\,. Thus, $(d_{\eta}f_{3_{\eta}})'(0.35)<L_{153_{\eta}}(0.4)$.
\item For $r_{5}\in(0.4,0.45]$\,. The upper and lower bounds are $(d_{\eta}f_{3_{\eta}})'(0.4)=3.19179$ and $L_{153_{0}}(0.45)=3.51233$, respectively\,. Thus, $(d_{\eta}f_{3_{\eta}})'(0.4)<L_{153_{\eta}}(0.45)$.
\item For $r_{5}\in(0.45,0.5]$\,. The upper and lower bounds are $(d_{\eta}f_{3_{\eta}})'(0.45)=2.66339$ and $L_{153_{0}}(0.5)=2.95406$, respectively\,. Thus, $(d_{\eta}f_{3_{\eta}})'(0.45)<L_{153_{\eta}}(0.5)$.
\end{itemize}
Therefore, we have $(d_{\eta}f_{3_{\eta}})'(r_{5})< L_{153_{\eta}}(r_{5})$.

\item Subinterval $r_{5}\in(0.5,b/2]$\,.  As with previous comparisons, we will now evaluate the upper and lower bounds for the families $(d_{\eta}f_{3_{\eta}})'(r_{5})$ and $ L_{154_{\eta}}(r_{5})$ within the subintervals of $r_{5}$: $(0.5,0.53]\cup(0.53,0.58]\cup(0.58,b/2]$\,. 
\begin{itemize}
\item For $r_{5}\in(0.5,0.53]$\,. The upper and lower bounds are  $(d_{\eta}f_{3_{\eta}})'(0.5)=2.21867$ and $L_{154_{0}}(0.53)=2.27436$, respectively. Thus, $(d_{\eta}f_{3_{\eta}})'(0.5)<L_{154_{\eta}}(0.53)$.

\item For $r_{5}\in(0.53,0.58]$\,.  The upper and lower bounds are  $(d_{\eta}f_{3_{\eta}})'(0.53)=1.9879$ and $L_{154_{0}}(0.58)=2.03677$, respectively. Thus, $(d_{\eta}f_{3_{\eta}})'(0.53)<L_{154_{\eta}}(0.58)$.

\item For $r_{5}\in(0.58,b/2]$\,.  The upper and lower bounds are  $(d_{\eta}f_{3_{\eta}})'(0.58)=1.65611$ and $L_{154_{0}}(b/2)=1.92885$, respectively. Thus, $(d_{\eta}f_{3_{\eta}})'(0.58)<L_{154_{\eta}}(b/2)$.
\end{itemize}
Therefore, $(d_{\eta}f_{3_{\eta}})'(r_{5})<L_{154_{\eta}}(r_{5})$.
\end{itemize}
Part I of the proof is concluded by demonstrating that $(d_{\eta}f_{3_{\eta}})'(r_{5})<-(d_{\eta}f_{2_{\eta}})'(r_{5})$.
\end{itemize}

\item \underline{Part II}:  For $\eta\in[0.037,\infty)$\,. We will divide the interval of $\eta$ into $[0.037,2.21)\cup[2.21,\infty)$.

\begin{itemize}
\item Part IIa: for $\eta\in[0.037,2.21)$\,. We will compare the upper and lower bounds of $(d_{\eta}f_{3_{\eta}})'(r_{5})$ and $-(d_{\eta}f_{2_{\eta}})'(r_{5})$\,. Both functions in these families are monotonically decreasing (Properties \ref{pfs})\,. Following the approach used in Part I, we will divide the interval for $r_{5}$ as follows $(0,0.2]\cup(0.2,0.3]\cup(0.3,0.4]\cup(0.4,0.5]\cup(0.5,0.6]\cup(0.6,b/2]$.
\begin{itemize}
\item For $r_{5}\in (0,0.2]$\,. The upper and lower bounds are  $(d_{\eta}f_{3_{0.037}})'(0)=8.72263$ and $-(d_{\eta}f_{2_{0.037}})'(0.2)=10.9563$, respectively.
\item For $r_{5}\in (0.2,0.3]$\,. The upper and lower bounds are  $(d_{\eta}f_{3_{0.037}})'(0.2)=5.91179$ and $-(d_{\eta}f_{2_{0.037}})'(0.3)=6.89814$, respectively. 
\item For $r_{5}\in (0.3,0.4]$\,. The upper and lower bounds are  $(d_{\eta}f_{3_{0.037}})'(0.3)=4.23557$ and $-(d_{\eta}f_{2_{0.037}})'(0.4)=4.48669$, respectively. 
\item For $r_{5}\in (0.4,0.5]$\,. The upper and lower bounds are  $(d_{\eta}f_{3_{0.037}})'(0.4)=2.94664$ and $-(d_{\eta}f_{2_{0.037}})'(0.5)=3.0142$, respectively. 
\item For $r_{5}\in (0.5,0.6]$\,. The upper and lower bounds are  $(d_{\eta}f_{3_{0.037}})'(0.5)=2.03089$ and $-(d_{\eta}f_{2_{0.037}})'(0.6)=2.08735$, respectively. 
\item For $r_{5}\in (0.6,b/2]$\,. The upper and lower bounds are  $(d_{\eta}f_{3_{0.037}})'(0.6)=1.40277$ and $-(d_{\eta}f_{2_{0.037}})'(b/2)=1.95944$, respectively. 
\end{itemize}

Clearly, for each subinterval of $r_{5}\in(0,b/2]$, it holds that $(d_{\eta}f_{3_{\eta}})'(r_{5})<-(d_{\eta}f_{2_{\eta}})'(r_{5})$.

\item Part IIb: for $\eta\in[2.21,\infty)$\,.  The functions of the families $-(d_{\eta}f_{2_{\eta}})'(r_{5})$ and $(d_{\eta}f_{3_{\eta}})'(r_{5})$ are monotonically decreasing\,. Therefore, the minimum values for the family $-(d_{\eta}f_{2_{\eta}})'(r_{5})$ are represented by $-(d_{\eta}f_{2_{\eta}})'(b/2)$ and the maximum values of the family $(d_{\eta}f_{3_{\eta}})'(r_{5})$ are represented by $(d_{\eta}f_{3_{\eta}})'(0)$\,. Both functions, $-(d_{\eta}f_{2_{\eta}})'(b/2)$ and  $(d_{\eta}f_{3_{\eta}})'(0)$, are continuous with respect to $\eta$\,. The function $-(d_{\eta}f_{2_{\eta}})'(b/2)$ is strictly decreasing, and its minimum value is given by  $\displaystyle{\lim_{\eta\rightarrow \infty}}(-(d_{\eta}f_{2_{\eta}})'(b/2))=b/2=0.618034$\,. On the other hand, the function $(d_{\eta}f_{3_{\eta}})'(0)$ is strictly decreasing for $\eta\in(2.21,4.88)$ and strictly increasing for $\eta\in[4.88,\infty)$,  with its maximum value at  $d_{\eta}f'_{3_{2.21}}(0)=0.562931$\,. Thus, $(d_{\eta}f_{3_{\eta}})'(r_{5})<-(d_{\eta}f_{2_{\eta}})'(r_{5})$\,. \end{itemize}
Part I and Part II complete the proof.
\end{proof}

\begin{proposition}\label{aux4}
For into $J_{1}$, the functions of the family $\left(d_{\eta}f_{1_{\eta}}\right)(r_{5})+\left(d_{\eta}f_{2_{\eta}}\right)(r_{5})+\left(d_{\eta}f_{3_{\eta}}\right)(r_{5})$ are convex.\end{proposition}

\begin{proof}
First of all, $\left(d_{\eta}f_{1_{\eta}}\right)(r_{5})=0$\,. We analyze the functions from the families $(d_{\eta}f_{2_{\eta}})''(r_{5})$ and $-(d_{\eta}f_{3_{\eta}})''(r_{5})$ in two parts\,. The interval for $\eta$  is divided into $(0,0.037]\cup(0.037,\infty)$.

\underline{Part I}: for $\eta\in(0,0.037]$\,. We divide the interval for $r_{5}$ into $(0,0.2]\cup(0.2,0.4]\cup(0.4,b/2]$.
\begin{itemize}
\item Parte Ia: for $r_{5}\in(0,0.2]$\,. According to the Properties \ref{spfs}, the lower bound of the family $(d_{\eta}f_{2_{\eta}})''(r_{5})$ is $(d_{\eta}f_{2_{0}})''(0.2)=49.6372$\,. Conversely, computational calculations show that the maximum values of the family $-(d_{\eta}f_{3_{\eta}})''(r_{5})$ form a strictly increasing sequence\,.  Thus, an upper bound is $-d_{\eta}f''_{3_{0.037}}(0.167719)=18.5569$\,. Therefore, $(d_{\eta}f_{2_{\eta}})''>-(d_{\eta}f_{3_{\eta}})''(r_{5})$.

\item Parte Ib: for $r_{5}\in(0.2,0.4]$\,. Both $(d_{\eta}f_{2_{\eta}})''(r_{5})$ and $-(d_{\eta}f_{3_{\eta}})''(r_{5})$ are monotonically decreasing (as stated in Properties \ref{spfs})\,. Thus, it suffices to compare $(d_{\eta}f_{2_{\eta}})''(0.4)$ and $-(d_{\eta}f_{3_{\eta}})''(0.2)$, which we can be viewed as functions of $\eta$\,. We will analyze these piecewise functions (Fig.~\ref{fig4}).

\begin{equation*}
    L_{156}(\eta):=\left\{ \begin{array}{lll}
        (d_{\eta}f_{2_{0}})''(0.2)=17.8795,  &0<\eta\leq 0.02,\\
             (d_{\eta}f_{2_{0.02}})''(0.4)=18.2481, & 0.02<\eta\leq 0.037.
               \end{array}
   \right.
\end{equation*}

\begin{equation*}
   L_{158}(\eta):=\left\{ \begin{array}{lll}
        -(d_{\eta}f_{3_{0.02}})''(0.2)=17.8526,  &0<\eta\leq 0.02,\\
             (d_{\eta}f_{3_{0.037}})''(0.2)=18.2362, & 0.02<\eta\leq 0.037.
               \end{array}
   \right. 
\end{equation*}

\begin{figure}[ht]
	\centering
		\includegraphics[scale=0.8]{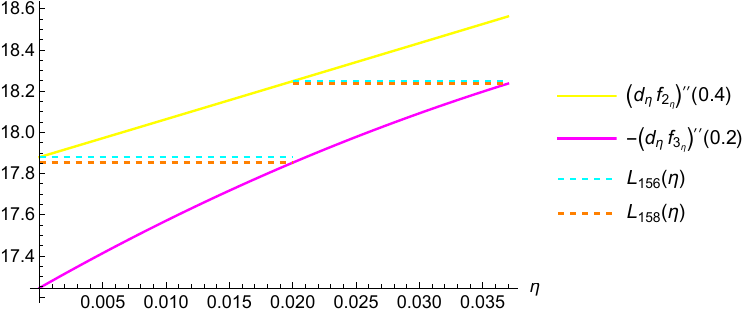}
	\caption{functions $(d_{\eta}f_{2_{\eta}})''(0.4)$, $-(d_{\eta}f_{3_{\eta}})''(0.2)$, $L_{156}(\eta)$ and $L_{158}(\eta)$.}
		\label{fig4}
\end{figure}

Clearly, $-(d_{\eta}f_{3_{\eta}})''(0.2)<L_{158}(\eta)<L_{156}(\eta)<(d_{\eta}f_{2_{\eta}})''(0.4)$\,. Therefore, it follows that $-(d_{\eta}f_{3_{\eta}})''(r_{5})<(d_{\eta}f_{2_{\eta}})''(r_{5})$\,. Hence, we have $(d_{\eta}f_{2_{\eta}})''(r_{5})+(d_{\eta}f_{3_{\eta}})''(r_{5})>0$.

\item Part Ic: for $r_{5}\in(0.4,b/2]$\,. The functions from the families $(d_{\eta}f_{2_{\eta}})''(r_{5})$ and $-(d_{\eta}f_{3_{\eta}})''(r_{5})$ are convex, as stated in 
Properties \ref{spfs}\,. To analyze these functions, we will consider the tangent lines to  $(d_{\eta}f_{2_{\eta}})''(r_{5})$ at $r_{5}=b/2$, denoted by $L_{160_{\eta}}(r_{5})$\,. We will also examine the lines connecting the endpoints of the functions from $-(d_{\eta}f_{3_{\eta}})''(r_{5})$ in the interval $r_{5}\in(0.4,b/2]$, denoted by $L_{163_{\eta}}(r_{5})$.

Calculations show that the family $L_{163_{\eta}}(r_{5})$ is strictly increasing\,. Similarly, the family of functions $L_{160_{\eta}}(r_{5})$ is strictly increasing, according to Properties \ref{spfs}\,. Thus, it suffices to compare $L_{160_{0}}(r_{5})$ with $L_{163_{0}}(r_{5})$\,. Both functions are monotonically decreasing, and the following inequalities hold:
\begin{equation*}
\begin{split}
L_{160_{0}}(0.4)=13.3123>L_{163_{0}}(0.4)=11.4622,\\
L_{160_{0}}(b/2)=6.62826>L_{163_{0}}(b/2)=5.0596.
\end{split}
\end{equation*}
Hence, $-(d_{\eta}f_{3_{\eta}})''(r_{5})<(d_{\eta}f_{2_{\eta}})''(r_{5})$.
\end{itemize}

\underline{Part II}: for $\eta\in(0.037,\infty)$\,. We divide the interval  $r_{5}\in(0,b/2]$ into two parts $(0,0.3]\cup(0.3,b/2]$\,. 

\begin{itemize}
\item Part IIa: for $r_{5}\in(0,0.3]$\,. Considering the growth of the functions of the family $(d_{\eta}f_{2_{\eta}})''(r_{5})$ and their monotonicity (Properties \ref{spfs}), we have a lower bound of $(d_{\eta}f_{2_{0.037}})''(0.3)$=30.8087\,. We further divide the interval of $\eta$ into $(0.037,0.26]\cup(0.26,\infty)$.

\begin{itemize}

\item For $\eta\in(0.037,0.26]$\,. Our computational calculations show that the maximum values of the family $-(d_{\eta}f_{3_{\eta}})''(r_{5})$ form a strictly increasing sequence\,. Therefore, an upper bound is $-d_{\eta}f''_{3_{0.26}}(0)=28.0721$\,. Hence $-(d_{\eta}f_{3_{\eta}})''(r_{5})<(d_{\eta}f_{2_{\eta}})''(r_{5})$.

\item For $\eta\in(0.26,\infty)$\,. According to Properties \ref{spfs}, the functions of the family $-(d_{\eta}f_{3_{\eta}})''(r_{5})$ are monotonically decreasing\,. Therefore, 
the upper bound is given by $-(d_{\eta}f_{3_{\eta}})''(0)$, with a maximum value of $-d_{\eta}f''_{3_{0.3575}}(0)=29.4401$\,. Hence $-(d_{\eta}f_{3_{\eta}})''(r_{5})<(d_{\eta}f_{2_{\eta}})''(r_{5})$.
\end{itemize}

\item Part IIb: for $r_{5}\in(0.3,b/2]$\,. The functions of the family  $-(d_{\eta}f_{3_{\eta}})''(r_{5})$ are convex\,.  We consider the family $L_{166_{\eta}}(r_{5})$ of lines connecting the endpoints of the graphs of these functions, which are monotonically decreasing\,. To compare $(d_{\eta}f_{2_{\eta}})''(r_{5})$ with  $L_{166_{\eta}}(r_{5})$, we evaluate different intervals for $r_{5}$, $(0.3,0.43]\cup(0.43,0.5]\cup(0.5,054]\cup(0.54,b/2]$\,. For each subinterval, we compare the maximum values of $L_{166_{\eta}}(r_{5})$ with the values of $(d_{\eta}f_{2_{0.037}})''(r_{5})$:
\begin{itemize}
\item In $r_{5}\in(0.3,0.43)$: the maximum value of $L_{166_{\eta}}(0.3)$ is $L_{166_{0.037}}(0.3)=14.9306$ and $(d_{\eta}f_{2_{0.037}})''(0.43)=16.0356$.
\item In $r_{5}\in(0.43,0.5)$: the maximum value of  $L_{166_{\eta}}(0.43)$  is $L_{166_{0.037}}(0.43)=10.6878$ and $(d_{\eta}f_{2_{0.037}})''(0.5)=11.5202$.
\item In $r_{5}\in(0.5,0.54)$: the maximum value of  $L_{166_{\eta}}(0.5)$ is $L_{166_{0.037}}(0.5)=18.40326$ and $(d_{\eta}f_{2_{0.037}})''(0.54)=9.60098$.
\item In $r_{5}\in(0.54,0.6)$: the maximum value of  $L_{166_{\eta}}(0.54)$ is $L_{166_{0.037}}(0.54)=7.09779$ and $(d_{\eta}f_{2_{0.037}})''(b/2)=7.3711$.
\end{itemize}
In all cases, $L_{166_{\eta}}(r_{5})<(d_{\eta}f_{2_{\eta}})''(r_{5})$\,. Thus, $-(d_{\eta}f_{3_{\eta}})''(r_{5})<(d_{\eta}f_{2_{\eta}})''(r_{5})$, which concludes the proof.
\end{itemize}
\end{proof}

\begin{proposition}\label{aux5}
For $r_{5}\in(0,b/2]$ and $\eta\in(0.02,\infty)$, the family of functions $\lambda_{31_{\eta}}(r_{5})$ is strictly increasing.\end{proposition}

\begin{proof}
Considering Proposition \ref{aux2} and the Properties \ref{es}, we know that the family 
$\lambda_{31\eta}(r_{5})$ is strictly increasing for \mbox{$\eta\in(7.60839,\infty)$}\,. Thus, we need to demonstrate the behavior of this family for $r_{5}\in(0,b/2]$ and $\eta\in(0.02,7.60839]$\,. Specifically, we will compare the functions $\left(d_{\eta}f_{1_{\eta}}\right)(r_{5})+\left(d_{\eta}f_{2_{\eta}}\right)(r_{5})+\left(d_{\eta}f_{3_{\eta}}\right)(r_{5})$ with $-\left(d_{\eta}f_{4_{\eta}}\right)(r_{5})$\,.
The proof will be conducted in three main parts, as outlined by the Properties of the family $\left(d_{\eta}f_{4_{\eta}}\right)(r_{5})$ described in \ref{dfs}, as showed in Fig. \ref{partes}.

\begin{figure}[ht]
	\centering
		\includegraphics[scale=0.36]{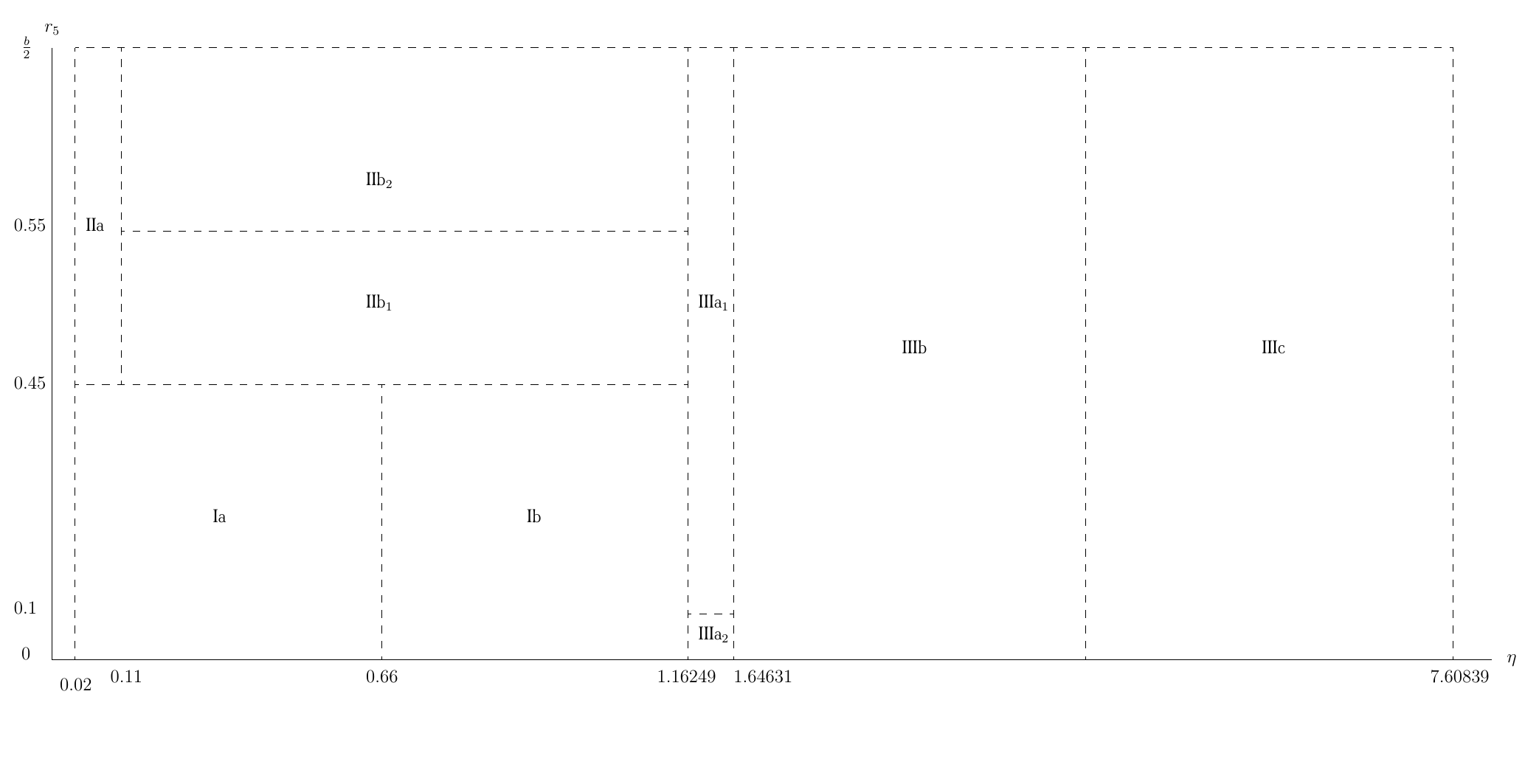}
	\caption{proof parts of Proposition \ref{aux5}.}
		\label{partes}
\end{figure}

\underline{Part I}\,. We will divide the interval for  $\eta$ into two parts $(0.02,0.66]\cup(0.66,1.16249]$.
\begin{itemize}
\item Part Ia\,. Computational analyses have established that the family of maximum values of $-\left(d_{\eta}f_{4_{\eta}}\right)(r_{5})$ is strictly increasing\,. Specifically, these maximum values are attained when $r_{5}\in(0.229163,0.45]$\,. Therefore, the analysis will focus solely on this interval\,. 
Additionally, according to Properties \ref{dfs}, the functions of the family $-\left(d_{\eta}f_{4_{\eta}}\right)$ are concave for $\eta$ into $(0.02,0.37]\cup(0.37,0.66]$\,. It is important to note that the functions of the family $\left(d_{\eta}f_{1_{\eta}}\right)(r_{5})+\left(d_{\eta}f_{2_{\eta}} \right)(r_{5})+\left(d_{\eta}f_{3_{\eta}}\right)(r_{5})$ exhibit decreasing monotonicity, as specified by Proposition \ref{aux3}\,. Therefore, for both intervals, we will use the family $\left(d_{\eta}f_{1_{\eta}}\right)(0.45)+\left(d_{\eta}f_{2_{\eta}} \right)(0.45)+\left(d_{\eta}f_{3_{\eta}}\right)(0.45)$.

\begin{itemize}
\item For $\eta\in(0.02,0.37]$\,. We will define the families $L_{9_{\eta}}(r_{5})$, $L_{10_{\eta}}(r_{5})$, and $L_{11_{\eta}}( r_{5})$ by constructing the tangent lines to the functions of the family $-\left(d_{\eta}f_{4_{\eta}}\right)$ at the points $r_{5}=0.229163$, $r_{5}=0.35$, and $r_{5}=0.45$, respectively (Fig.~\ref{fig6}).

\begin{figure}[ht]
	\centering
		\includegraphics[scale=0.76]{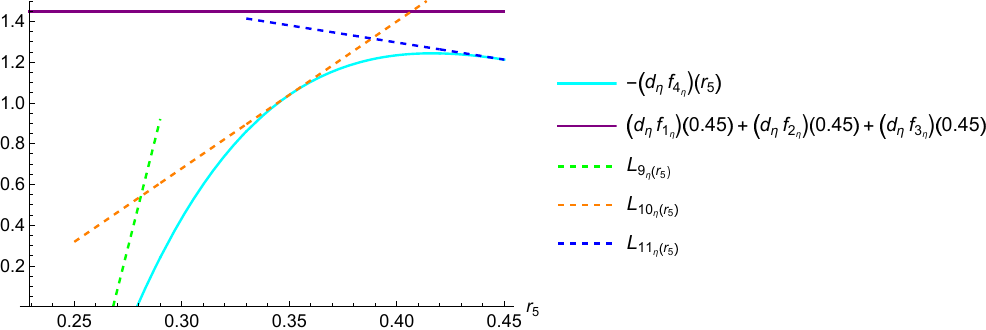}
	\caption{functions $-\left(d_{\eta}f_{4_{\eta}}\right)(r_{5})$, $\left(d_{\eta}f_{1_{\eta}}+d_{\eta}f_{2_{\eta}}+d_{\eta}f_{3_{\eta}}\right)(0.45)$, $L_{9_{\eta}}(r_{5})$, $L_{10_{\eta}}(r_{5})$ and $L_{11_{\eta}}(r_{5})$ with $\eta=0.1785$.}
		\label{fig6}
\end{figure}

\begin{itemize}
\item For $r_{5}\in(0.229163,0.35]$\,. Consider the intersection points of the line families $L_{9_{\eta}}(r_{5})$ and $L_{10_{\eta}}(r_{5})$ which constitute the family $p_{1}(\eta)$\,. Notably, the function $L_{9_{\eta}}(p_{1}(\eta))$ is positive for $\eta\in(0.1188,0.37)$\,. We will now compare the families $L_{9_{\eta}}(p_{1}(\eta))$ and $\left(d_{\eta}f_{1_{\eta}}\right)(0.45)+ \left(d_{\eta}f_{2_{\eta}}\right)(0.45)+\left(d_{\eta}f_{3_{\eta}}\right)(0.45)$, which are functions depending of $\eta$\,. Both functions are monotonically increasing and can be bounded by the following piecewise functions,

{\footnotesize
\begin{align*}
L_{12}(\eta)\colonequals
\begin{cases}
(d_{\eta}f_{1_{0.1188}}+d_{\eta}f_{2_{0.1188}}+d_{\eta}f_{3_{0.1188}})(0.45)=1.3443,&0.1188\leq\eta<0.3\\[2mm]
(d_{\eta}f_{1_{0.3}}+d_{\eta}f_{2_{0.3}}+d_{\eta}f_{3_{0.3}})(0.45)=1.62225,& 0.3\leq\eta<0.37.
\end{cases}
\end{align*}}

\begin{align*}
L_{13}(\eta):= \left\{ \begin{array}{lll}
          L_{9_{0.3}}(pn_{1}(0.3))=1.30223, & 0.1188\leq\eta<0.3\\
          L_{9_{0.37}}(pn_{1}(0.37))=1.56241, & 0.3\leq\eta<0.37.
\end{array}
   \right.
\end{align*}

Clearly, $L_{9_{\eta}}(pn_{1}(\eta))<L_{13}(\eta)<L_{12}(\eta)<\left(d_{\eta}f_{1_{\eta}}+d_{\eta}f_{2_{\eta}}+d_{\eta}f_{3_{\eta}}\right)(0.45)$, (see Fig.~\ref{fig7}).

\begin{figure}[ht]
	\centering
		\includegraphics[scale=0.75]{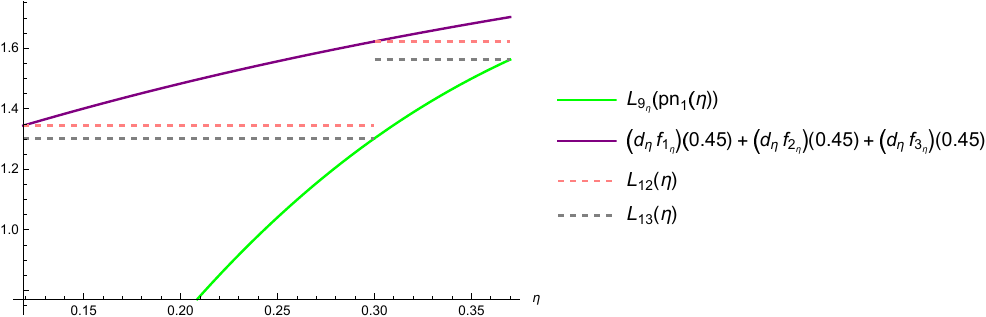}
	\caption{functions $L_{9_{\eta}}(pn_{1}(\eta))$, $\left(d_{\eta}f_{1_{\eta}}+d_{\eta}f_{2_{\eta}}+d_{\eta}f_{3_{\eta}}\right)(0.45)$, $L_{12}(\eta)$ and $L_{13}(\eta)$\,.}
		\label{fig7}
\end{figure}

\item For $r_{5}\in(0.35,0.45]$\,. We proceed similarly to the previous analysis\,. Let $p_{2}(\eta)$ denote the family of intersection points where the line families $L_{10_{\eta}}(r_{5})$ and $L_{11_{\eta}}(r_{5})$ intersect\,. We will compare the families $L_{10_{\eta}}(p_{2}(\eta))$ with $\left(d_{\eta}f_{1_{\eta}}\right)(0.45)+ \left(d_{\eta}f_{2_{\eta}}\right)(0.45)+\left(d_{\eta}f_{3_{\eta}}\right)(0.45)$\,. Both of these families can be viewed as functions of $\eta$ and are monotonically increasing\,. Therefore, we introduce the next piecewise functions, as shown in Fig.~\ref{fig8},

\begin{figure}[ht]
	\centering
		\includegraphics[scale=0.8]{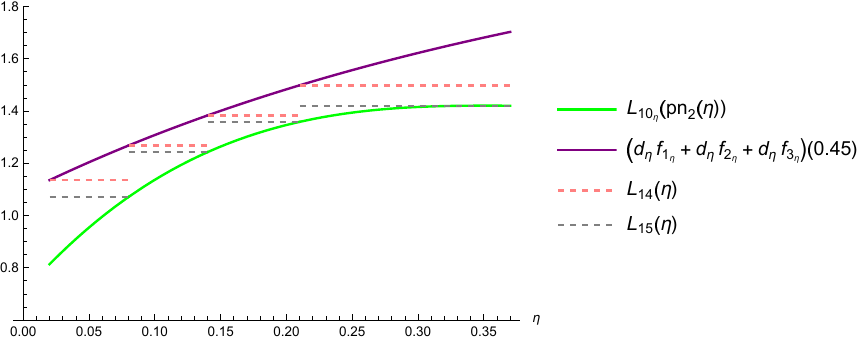}
	\caption{functions $L_{10_{\eta}}(p_{2}(\eta))$, $\left(d_{\eta}f_{1_{\eta}}\right)(0.45)+ \left(d_{\eta}f_{2_{\eta}}\right)(0.45)+\left(d_{\eta}f_{3_{\eta}}\right)(0.45)$, $L_{14}(\eta)$ and $L_{15}(\eta)$\,.}
		\label{fig8}
\end{figure}

{\footnotesize
\begin{align*}
L_{14}(\eta):= \left\{ \begin{array}{lll}
            (d_{\eta}f_{1_{0.02}}+d_{\eta}f_{2_{0.02}}+d_{\eta}f_{3_{0.02}})(0.45)=1.13542, & 0.02\leq\eta<0.08,\\[2mm]
         (d_{\eta}f_{1_{0.08}}+d_{\eta}f_{2_{0.08}}+d_{\eta}f_{3_{0.08}})(0.45)=1.26796 , & 0.08\leq\eta<0.14,\\ [2mm]       
         (d_{\eta}f_{1_{0.14}}+d_{\eta}f_{2_{0.14}}+d_{\eta}f_{3_{0.14}})(0.45)=1.38307 , & 0.14\leq\eta<0.21,\\ [2mm]        
          (d_{\eta}f_{1_{0.21}}+d_{\eta}f_{2_{0.21}}+ d_{\eta}f_{3_{0.21}})(0.45)=1.38307 , & 0.21\leq\eta<0.37,
\end{array}
   \right. 
\end{align*}
}
\begin{equation*}
L_{15}(\eta):= \left\{ \begin{array}{lll}
          L_{10_{0.08}}(pn_{1}(0.08))=1.071, & 0.02\leq\eta<0.08,\\
            
          L_{10_{0.14}}(pn_{1}(0.14))=1.24254, & 0.08\leq\eta<0.14,\\
          
          L_{10_{0.21}}(pn_{1}(0.21))=1.35853, & 0.14\leq\eta<0.21,\\
          
          L_{10_{0.37}}(pn_{1}(0.37))=1.42032, & 0.21\leq\eta<0.37,
\end{array}
   \right.
\end{equation*}
and clearly, the inequality $L_{10_{\eta}}(pn_{2}(\eta))<L_{15}(\eta)<L_{14}(\eta)<\left(d_{\eta}f_{1_{\eta}}+d_{\eta}f_{2_{\eta}}+d_{\eta}f_{3_{\eta}}\right)(0.45)$ holds.
\end{itemize}

\item For $\eta\in(0.37,0.66]$\,. We analyze a lower bound of the expression $\left(d_{\eta}f_{1_{\eta}}\right)(0.45)+\left(d_{\eta}f_{2_{\eta}}\right) (0.45)+\left(d_{\eta}f_{3_{\eta}}\right)(0.45)$ and an upper bound of $-\left(d_{\eta}f_{4_{\eta }}\right)(r_{5})$.The interval for $\eta$ is divided into $(0.37,0.53]\cup(0.53,0.63]\cup(0.63,0.66]$\,. The following conclusions are drawn: 
\begin{equation*}
\begin{split}
(d_{\eta}f_{1_{0.37}}+d_{\eta}f_{2_{0.37}}+d_{\eta}f_{3_{0.37}})(0.45)=1.70293&>-(d_{\eta}f_{4_{0.53}})(0.272695)=1.67685,\\
(d_{\eta}f_{1_{0.53}}+d_{\eta}f_{2_{0.53}}+d_{\eta}f_{3_{0.53}})(0.45)=1.84892&>-(d_{\eta}f_{4_{0.63}})(0.238827)=1.81203, \\ 
(d_{\eta}f_{1_{0.63}}+d_{\eta}f_{2_{0.63}}+d_{\eta}f_{3_{0.63}})(0.45)=1.9199&>-(d_{\eta}f_{4_{0.66}})(0.229163)=1.85361.
\end{split}
\end{equation*}

Therefore, it follows that $\left(d_{\eta}f_{1_{\eta}}+d_{\eta}f_{2_{\eta}}+d_{ \eta}f_{3_{\eta}}\right)(r_{5})>-\left(d_{\eta}f_{4_{\eta}}\right)(r_{5})$, as illustrated in Figures~\ref{fig9},~\ref{fig10} and~\ref{fig11}.

\begin{figure}[ht]
	\centering
		\includegraphics[scale=0.8]{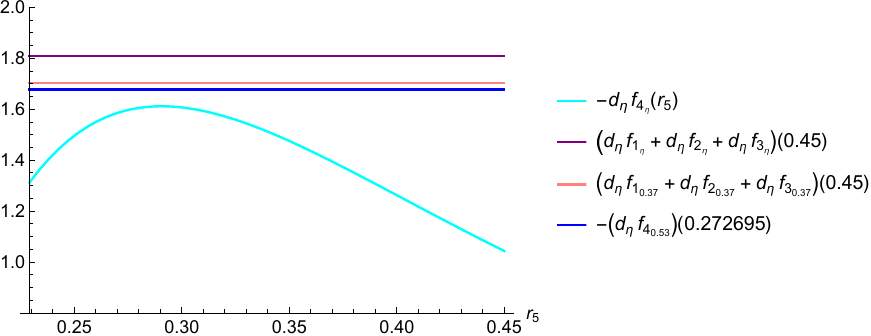}
	\caption{functions $-\left(d_{\eta}f_{4_{\eta}}\right)(r_{5})$, $\left(d_{\eta}f_{1_{\eta}}+d_{\eta}f_{2_{\eta}}+d_{\eta}f_{3_{\eta}}\right)(0.45)$, $(d_{\eta}f_{1_{0.37}}+d_{\eta}f_{2_{0.37}}+d_{\eta}f_{3_{0.37}})(0.45)$ and $-(d_{\eta}f_{4_{0.53}})(0.272695)$ with $\eta=0.48$\,.}
		\label{fig9}
\end{figure}

\begin{figure}[ht]
	\centering
		\includegraphics[scale=0.8]{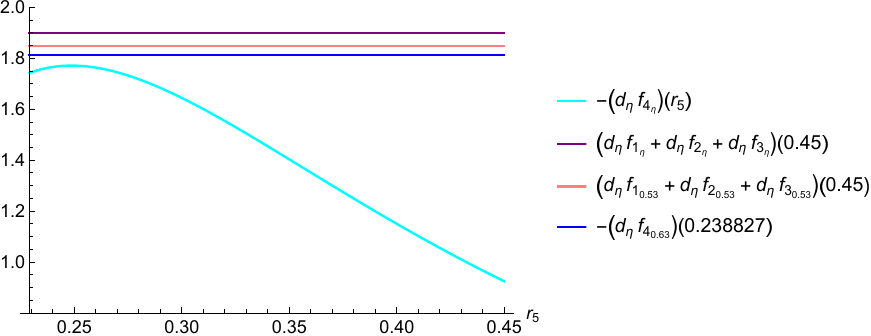}
	\caption{functions $-\left(d_{\eta}f_{4_{\eta}}\right)(r_{5})$, $\left(d_{\eta}f_{1_{\eta}}+d_{\eta}f_{2_{\eta}}+d_{\eta}f_{3_{\eta}}\right)(0.45)$, $(d_{\eta}f_{1_{0.53}}+d_{\eta}f_{2_{0.53}}+d_{\eta}f_{3_{0.53}})(0.45)$ and $-(d_{\eta}f_{4_{0.53}})(0.238827)$, with $\eta=0.6$\,.}
		\label{fig10}
\end{figure}

\begin{figure}[ht]
	\centering
		\includegraphics[scale=0.8]{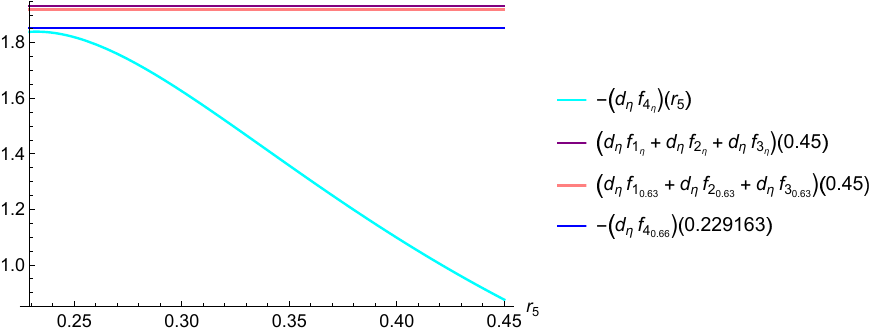}
	\caption{functions $-\left(d_{\eta}f_{4_{\eta}}\right)(r_{5})$, $\left(d_{\eta}f_{1_{\eta}}+d_{\eta}f_{2_{\eta}}+d_{\eta}f_{3_{\eta}}\right)(0.45)$ and $(d_{\eta}f_{1_{0.63}}+d_{\eta}f_{2_{0.63}}+d_{\eta}f_{3_{0.63}})(0.45)$, with $\eta=0.65$\,.}
		\label{fig11}
\end{figure}
\end{itemize}

\item Part Ib: for $\eta\in(0.66,1.16249]$\,. We will analyze the interval for $r_{5}$ by dividing it into $(0,0.33]\cup(0.33 .0.45]$.
\begin{itemize}
\item For $r_{5}\in(0,0.33]$\,. Given that the family of maxima for $-\left(d_{\eta}f_{4_{\eta}}\right)(r_{5})$ is strictly increasing, and that the family $\left(d_{ \eta}f_{1_{\eta}}+d_{\eta}f_{2_{\eta}}+d_{\eta}f_{3_{\eta}}\right)(r_{5})$ is also strictly increasing (Properties \ref{dfs}), we will compare the maximum value of one family with the minimum value of the other\,. Specifically, the maximum value of $-\left(d_{\eta}f_{4_{\eta}}\right)(r_{5})$ is reached at \mbox{$\eta = 1.16249$} and corresponds to $-(d_{\eta}f_{4_{1.16249}})(0.0945512) = 2.62007$\,. This value is less than the minimum value of $\left(d_{\eta}f_{1_{\eta}} + d_{\eta}f_{2_{\eta}} + d_{\eta}f_{3_{\eta}}\right)(r_{5})$ which is reached at $\eta = 0.66$ and it is $\left(d_{\eta}f_{1_{0.66}} + d_{\eta}f_{2_{0.66}} + d_{\eta}f_{3_{0.66}}\right)(0.33) = 2.6338$, (see Fig.~\ref{fig12}).

\item For $r_{5}\in(0.33,0.45]$\,. A lower bound for $\left(d_{\eta}f_{1_{\eta}}+d_{\eta}f_{2_{\eta}}+d_{\eta}f_{3_{\eta}}\right)(r_{5})$ is given by $d_{\eta}f_{1_{0.66}}(0.45)+d_{\eta}f_{2_{0.66}} (0.45)+d_{\eta}f_{3_{0.66}}(0.45)=1.93885$\,. Conversely, an upper bound for 
$-\left(d_{\eta}f_{4_{\eta}}\right)(r_ {5})$ is $-(d_{\eta}f_{4_{0.66}})(0.33)=1.4581$, (see Fig.~\ref{fig13}).

\begin{figure}[ht]
	\centering
		\includegraphics[scale=0.8]{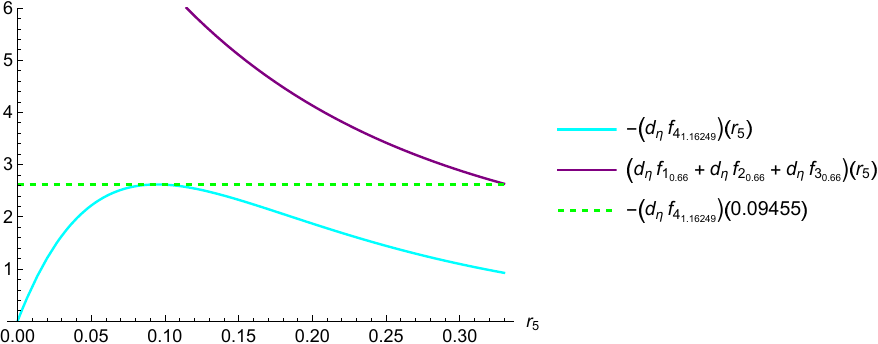}
	\caption{functions $-d_{\eta}f_{4_{1.16249}}(r_{5})$, $d_{\eta}f_{1_{0.66}}(r_{5})+d_{\eta}f_{2_{0.66}}(r_{5})+d_{\eta}f_{3_{0.66}}(r_{5})$ and $-d_{\eta}f_{4_{1.16249}}(0.09455)$, with $\eta=1.16249$\,.}
		\label{fig12}
\end{figure}

\begin{figure}[ht]
	\centering
		\includegraphics[scale=0.8]{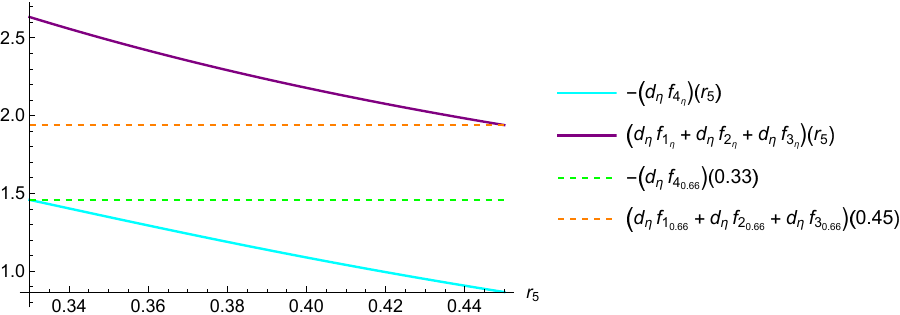}
	\caption{functions $-\left(d_{\eta}f_{4_{\eta}}\right)(r_{5})$, $\left(d_{\eta}f_{1_{\eta}}+d_{\eta}f_{2_{\eta}}+d_{\eta}f_{3_{\eta}}\right)(r_{5})$ with $\eta=0.66$, and their respective upper and lower bounds\,.}
		\label{fig13}
\end{figure}

\end{itemize}
Therefore, its true that $-\left(d_{\eta}f_{4_{\eta}}\right)(r_{5})<\left(d_{\eta}f_{1_{\eta}}+d_{ \eta}f_{2_{\eta}}+d_{\eta}f_{3_{\eta}}\right)(r_{5})$.
\end{itemize}
\underline{Part II}\,. We divide the interval for $\eta$ into two intervals $(0.02,0.11]\cup(0.11,1.16249]$.
\begin{itemize}
\item Part IIa: for $\eta\in(0.02,0.11]$\,. Based on Properties \ref{dfs}, the functions $-\left(d_{\eta}f_{4_{\eta}}\right)(r_{5})$ exhibit concavity for $r_{5} \in (0.45, 0.582276]$\,. Consequently, we partition the interval for $r_{5}$ into two subintervals: $(0.45, 0.582276]$ and $(0.582276, b/2]$.
\begin{itemize}
\item For $r_{5}\in(0.45, 0.582276]$\,. We consider the families $L_{28_{\eta}}(r_{5})$, $L_{29_{\eta}}(r_{5})$, $L_{30_{\eta}}(r_{5})$, and $L_{31_{\eta}}(r_{5})$, which are constructed from the tangent lines to the functions in the family $-\left(d_{\eta}f_{4_{\eta}}\right)(r_{5})$ at the points $r_{5}=0.45$, $r_{5}=0.48$, $r_{5}=0.5$ and $r_{5}=0.52$, respectively\,. Since the functions in the  family $\left(d_{\eta}f_{1_{\eta}}+d_{ \eta}f_{2_{\eta}}+d_{\eta}f_{3_{\eta}}\right)(r_{5})$ are monotonically decreasing (Proposition \ref{aux3}), we will perform the following comparisons:

\begin{itemize}
\item Interval $r_{5}\in(0.45,0.48]$: We compare $L_{28_{\eta}}(r_{5})$ with $\left(d_{\eta}f_{1_{\eta}}+d_{ \eta}f_{2_{\eta}}+d_{\eta}f_{3_{\eta}}\right)(0.48)$\,. The family of lines, $L_{28_{\eta}}(r_{5})$, exhibits monotonic increasing behavior\,. To compare $L_{28_{\eta}}(0.48)$ and $\left(d_{\eta}f_{1_{\eta}}+d_{\eta}f_{2_{\eta}}+d_{\eta}f_{3_{\eta}}\right)(0.48)$, consider these as functions of $\eta$\,. Upon simple calculation, we find that both functions are concave\,. We then create the lines $L_{32}(\eta)$, which connect the endpoints of the curve $\left(d_{\eta}f_{1_{\eta}}+d_{\eta}f_{2_{\eta}}+d_{\eta}f_{3_{\eta}}\right)(0.48)$, and the line $L_{33}(\eta)$, which is tangent to $L_{28_{\eta}}(0.48)$ at the point where $\eta=0.02$\,. The following information is accurate:
\begin{equation*}
\begin{split}
L_{32}(0.02)=1.10507 \quad \mbox{and} \quad L_{33}(0.02)=1.09702,\\
L_{32}(0.11)=1.27808 \quad \mbox{and} \quad L_{33}(0.11)=1.22348.
\end{split}
\end{equation*}

In this way, $L_{32}(r_{5})>L_{33}(r_{5})$.

\item Interval $r_{5}\in(0.48,0.5]$: We will compare $L_{29_{\eta}}(r_{5})$ with $\left(d_{\eta}f_{1_{\eta}}+d_{ \eta}f_{2_{\eta}}+d_{\eta}f_{3_{\eta}}\right)(0.5)$\,. For $\eta\in(0.02,0.052344]$, the lines in the family $L_{29_{\eta}}(r_{5})$  exhibit increasing monotonicity, whereas for $\eta\in (0.052444,0.11]$, they display decreasing monotonicity.

\underline{For $\eta\in(0.02,0.052344]$:} We compare the values of $L_{29_{\eta}}(0.5)$ with $\left(d_{\eta}f_{1_{\eta}}+d_{\eta}f_{2_{\eta}}+d_{\eta}f_{3_{\eta}}\right)(0.5)$\,. Both functions are concave\,. We construct a line, $L_{34}(\eta)$, that connects the endpoints of $\left(d_{\eta}f_{1_{\eta}}+d_{\eta}f_{2_{\eta}}+d_{\eta}f_{3_{\eta}}\right)(0.5)$ and the tangent line to $L_{29_{\eta}}(0.5)$ at $\eta=0.02$\,. Both lines are monotonically increasing\,. The following values are noted:
\begin{equation*}
\begin{split}
L_{34}(0.02)=1.0869 \quad &\mbox{and} \quad L_{35}(0.02)=1.08011,\\
L_{34}(0.052344)=1.1491 \quad &\mbox{and} \quad L_{35}(0.052344)=1.11022.
\end{split}
\end{equation*}
This comparison demonstrates that the desired condition is met.

\underline{For $\eta\in(0.052344,0.11]$:} To compare the values of the family of lines \(L_{29_{\eta}}(r_{5})\), consider \(L_{29_{\eta}}(0.48)\) and \(\left(d_{\eta}f_{1_{\eta}}+d_{\eta}f_{2_{\eta}}+d_{\eta}f_{3_{\eta}}\right)(0.5)\) as functions of \(\eta\)\,. Both functions are monotonically increasing\,. The upper bound of \(L_{29_{\eta}}(0.48)\) is \(L_{29_{0.11}}(0.48)=1.1441\), while the lower bound of \(\left( d_{\eta}f_{1_{\eta}}+d_{\eta}f_{2_{\eta}}+d_{\eta}f_{3_{\eta}}\right)(0.5)\) is \(\left(d_{\eta}f_{1_{0.052344}}+d_{\eta}f_{2_{0.052344}}+d_{\eta}f_{3_{0.052344}}\right)(0.5)=1.1491\)\,.  This completes the comparison.

\item Interval $r_{5}\in(0.5,0.52]$: We compare $L_{30_{\eta}}(r_{5})$ with $\left(d_{\eta}f_{1_{\eta}}+d_{ \eta}f_{2_{\eta}}+d_{\eta}f_{3_{\eta}}\right)(0.52)$\,. The family of functions $L_{30_{\eta}}(r_{5})$ is known to be decreasing\,. Consequently, we compare $L_{30_{\eta}}(0.5)$ with $\left(d_{\eta}f_{1_{\eta}}+d_{ \eta}f_{2_{\eta}}+d_{\eta}f_{3_{\eta}}\right)(0.52)$, both of which are functions of $\eta$\,. Let $L_{50}(\eta)$ denote the tangent line to $L_{30_{\eta}}(0.5)$ at $\eta=0.02$, and let $L_{51}(\eta)$ represent the straight line connecting the endpoints of $\left(d_{\eta}f_{1_{\eta}}+d_{ \eta}f_{2_{\eta}}+d_{\eta}f_{3_{\eta} }\right)(0.52)$\,. Both $L_{50}(\eta)$ and $L_{51}(\eta)$ are monotonically increasing\,. The values are as follows: 
\begin{equation*}
\begin{split}
L_{50}(0.02)=1.06881 \quad \mbox{and} \quad L_{51}(0.02)=1.07015,\\
L_{50}(0.11)=1.15834 \quad \mbox{and} \quad L_{51}(11)=1.22094.
\end{split}
\end{equation*}
This comparison meets the desired objective.

\item  Interval $r_{5}\in(0.52,0.582276]$\,. We compare $L_{31_{\eta}}(r_{5})$ with $\left(d_{\eta}f_{1_{\eta}}+d_{ \eta}f_{2_{\eta}}+d_{\eta}f_{3_{\eta}}\right)(r_{5})$\,.  Note the following: let $L_{36_{\eta}}(r_{5})$ represent the family of tangent lines to the functions $\left(d_{\eta}f_{1_{\eta}}+d_{\eta}f_{2_{\eta}}+d_{\eta}f_{3_{\eta}}\right)(r_{5})$ at $r_{5}=0.52$\,. We need to compare $L_{36_{\eta}}(r_{5})$ with $L_{31_{\eta}}(r_{5})$\,. Both families of lines are monotonically decreasing, so we compare $L_{31_{\eta}}(0.52)$ and $L_{36_{\eta}}(0.52)$, as well as $L_{31_{\eta}}(0.582276)$ and $L_{36_{\eta}}(0.582276)$.

\underline{Families $L_{31_{\eta}}(0.52)$ y $L_{36_{\eta}}(0.52)$:} We divide the interval for $\eta$ into $(0.02,0.04]\cup(0.04,0.11]$.

    ** For $\eta\in(0.02,0.04]$\,. Both families $L_{31_{\eta}}(0.52)$ and $L_{36_{\eta}}(0.52)$ are strictly increasing\,. Thus, an upper bound for $L_{31_{\eta}}(0.52 )$ is $L_{31_{0.04}}(0.52)=1.06649$, and a lower bound for $L_{36_{\eta}}(0.52)$ is $L_{36_{0.04}}(0.52)=1.07015 $\,. Therefore, $L_{31_{\eta}}(0.52)<L_{36_{\eta}}(0.52)$.
        
    ** For $\eta\in(0.04,0.11]$\,. We treat $L_{31_{\eta}}(0.52)$ and $L_{36_{\eta}}(0.52)$ as functions of $\eta$\,. A simple calculation show that the maximum value of $L_{31_{\eta}}(0.52)$ is $L_{31_{0.074812}}(0.52)=1.07179$, while a lower bound for $L_{36_{\eta}}(0.52)$ is $L_{36_{0.04}}(0.52)=1.10666$\,. Thus, $L_{31_{\eta}}(0.52)<L_{36_{\eta}}(0.52)$.
    
\underline{Families $L_{31_{\eta}}(0.582276)$ and $L_{36_{\eta}}(0.582276)$:} The family $L_{31_{\eta}}(0.582276)$ is strictly decreasing with its maximum value being $L_{31_{0.02}}(0.582276)=1.00089$\,. In contrast, $L_{36_{\eta}}(0.582276)$ is strictly increasing with its minimum value being $L_{36_{0.02}}(0.582276)=1.02002$\,. Consequently, 
$L_{31_{\eta}}(0.582276)<L_{36_{\eta}}(0.582276)$.
\end{itemize}

\item For $r_{5}\in(0.582276,b/2]$\,. Considering the Properties stated in \ref{dfs}, both families of functions are monotonically decreasing\,. Therefore, it suffices to compare the values of $\left(d_{\eta}f_{1_{\eta}}+d_{\eta}f_{2_{\eta}}+d_{\eta}f_{3_{\eta}}\right)(b/2)$ and $-\left(d_{\eta}f_{4_{\eta}}\right)(0.582276)$, which are functions of $\eta$\,. The  function $\left(d_{\eta}f_{1_{\eta}}+d_{\eta}f_{2_{\eta}}+d_{\eta}f_{3_{\eta}}\right) (b/2)$ is monotonically increasing, with its minimum value given by $\left(d_{\eta}f_{1_{0.02}}+d_{\eta}f_{2_{0.02}}+d_{\eta} f_{3_{0.02}}\right)(b/2)=1.00361$\,. Conversely, the function $-\left(d_{\eta}f_{4_{\eta}}\right)(0.582276)$ is monotonically decreasing with its maximum value being $-\left(d_{\eta} f_{4_{0.02}}\right)(0.582276)=0.953404$\,. Consequently, the inequality $-\left(d_{\eta}f_{4_{\eta}}\right)(r_{5})<\left(d_{\eta}f_{1_{\eta}}+d_{\eta }f_{2_{\eta}}+d_{\eta}f_{3_{\eta}}\right)(r_{5})$ its true.
\end{itemize}

\item Part IIb: for $\eta\in(0.11,1.16249]$\,. Based on the Properties \ref{dfs}, we can see that the functions $-\left(d_{\eta}f_{4_{\eta}}\right)(r_{5})$ and $\left(d_{\eta}f_{1_{\eta}}+d_{\eta}f_{2_{\eta}}+d_{\eta}f_{3_{\eta}}\right)(r_{5})$ are both decreasing\,. We will divide the interval of $r_{5}$ into two subintervals: $(0.45,0.55]\cup(0.55,b/2]$\,. In each of these subintervals, we will compare the family of upper bounds of $-\left(d_{\eta}f_{4_{\eta}}\right)(r_{5})$ with the lower bounds of $\left(d_{\eta}f_{1_{\eta}}+d_{\eta}f_{2_{\eta}}+d_{\eta}f_{3_{\eta}}\right)(r_{5})$.

\begin{itemize}
\item Parte IIb$_{1}$: for $r_{5}\in(0.45,0.55]$\,. Let us consider the families $\left(d_{\eta}f_{1_{\eta}}+d_{\eta}f_{2_{\eta}}+d_{\eta}f_{3_{\eta }}\right)(0.55)$ and $-\left(d_{\eta}f_{4_{\eta}}\right)(0.45)$, which depend on the variable $\eta$\,. With some simple calculations, we find that the function $-\left(d_{\eta}f_{4_{\eta}}\right)(0.45)$ is concave for $\eta\in(0.11,0.557561]$, convex for $\eta\in(0.557561,1.16249]$, and monotonically decreasing within the latter interval\,. In contrast, the function $\left(d_{\eta}f_{1_{\eta}}+d_{\eta}f_{2_{\eta}}+d_{\eta}f_{3_{\eta}}\right)(0.55)$ is concave.
\begin{itemize}
\item For $\eta\in(0.11,0.557561]$\,. We consider the function $L_{24}(\eta)$, which is the line that connects the endpoints of the function $\left(d_{\eta}f_{1_{\eta}}+d_{\eta}f_{2_{\eta} }+d_{\eta}f_{3_{\eta}}\right)(0.55)$\,. We also consider the lines $L_{25}(\eta)$ and $L_{26}(\eta)$, which are tangent lines to $-\left(d_{\eta}f_{4_{\eta}}\right)(0.45)$ at the points $\eta=0.11$ and $\eta=0.15$, respectively as showed in Fig.~\ref{fig14}.

\begin{figure}[ht]
	\centering
		\includegraphics[scale=0.8]{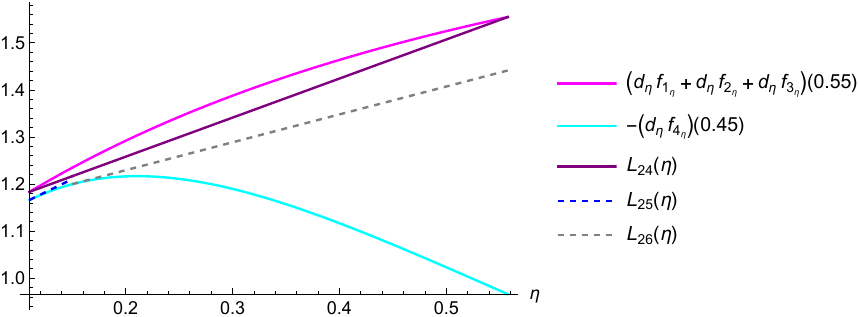}
	\caption{functions $\left(d_{\eta}f_{1_{\eta}}+d_{\eta}f_{2_{\eta}}+d_{\eta}f_{3_{\eta}}\right)(0.55)$, $-\left(d_{\eta}f_{4_{\eta}}\right)(0.45)$, $L_{24}(\eta)$, $L_{25}(\eta)$ and $L_{26}(\eta)$.}
		\label{fig14}
\end{figure}
The lines $L_{24}(\eta)$, $L_{25}(\eta)$ and $L_{26}(\eta)$ are increasing monotone and the following holds:
\begin{equation*}
\begin{split}
L_{24}(0.11)&=1.18358 \quad \mbox{ and } \quad L_{25}(0.11)=1.16632,\\
L_{24}(0.15)&=1.21676, \quad L_{25}(0.15)=1.21054 \quad \mbox{ and } \quad L_{26}(0.15)=1.19992,\\
L_{24}(b/2)&=1.60502 \quad \mbox{ and } \quad L_{26}(b/2)=1.14776.
\end{split}
\end{equation*}
It is true that $L_{24}(\eta)>L_{25}(\eta)$ y $L_{24}(\eta)>L_{26}(\eta)$\,. Hence, $\left(d_{\eta}f_{1_{\eta}}+d_{\eta}f_{2_{\eta} }+d_{\eta}f_{3_{\eta}}\right)(0.55)>-\left(d_{\eta}f_{4_{\eta}}\right)(0.45)$.

\item For $\eta\in(0.557561,1.16249]$\,. We consider the line $L_{27}(\eta)$, which connects the ends of the curve given by the function $\left(d_{\eta}f_{1_{\eta}}+d_{\eta}f_{2_{\eta} }+d_{\eta}f_{3_{\eta}}\right)(0.55)$\,. It is easy to see that such line is monotonically increasing, therefore, the minimum value is $L_{27}(0.557561)=1.55486$\,. On the other hand, the maximum value of $-\left(d_{\eta}f_{4_{\eta}}\right)(0.45)$ is $-\left(d_{\eta}f_{4_{0.557561} }\right)(0.45)=0.966289$\,. Then, $\left(d_{\eta}f_{1_{\eta}}+d_{\eta}f_{2_{\eta} }+d_{\eta}f_{3_{\eta}}\right) (0.55)>-\left(d_{\eta}f_{4_{\eta}}\right)(0.45)$.
\end{itemize}
\item Parte IIb$_{2}$: for $r_{5}\in(0.55,b/2]$\,. We compare the families $-\left(d_{\eta}f_{4_{\eta}}\right)(0.55)$ and $\left(d_{\eta}f_{1_{\eta}}+d_{\eta}f_{2_{\eta} }+d_{\eta}f_{3_{\eta}}\right)(b/2)$, which can be viewed as functions dependent on $\eta$\,. It is easy to see that the function $\left(d_{\eta}f_{4_{\eta}}\right)(0.55)$ is strictly decreasing, reaching its maximum value at $-\left(d_{\eta}f_{4_{0.11}}\right)(0.55)=0.99126$, while the function $\left(d_{\eta}f_{1_{\eta}}+d_{\eta}f_{2_{\eta} }+d_{\eta}f_{3_{\eta}}\right)(b/2)$ is strictly decreasing and attains its minimum value at $\left(d_{\eta}f_{1_{0.11}}+d_{\eta}f_{2_{0.11}}+d_{\eta}f_{3_{0.11}}\right)(b/2)=1.11263$\,. Then $-\left(d_{\eta}f_{4_{\eta}}\right)(0.55)<\left(d_{\eta}f_{1_{\eta}}+d_{\eta}f_{2_{\eta} }+d_{\eta}f_{3_{\eta}}\right)(b/2)$.
\end{itemize}
\end{itemize}
\underline{Part III}\,. First, let us partition the interval $\eta$ into three segments: $(1.16249,1.64631]\cup(1.64631,4.52952]\cup(4.52952,7.60839]$\,. 
\begin{itemize}
\item Part IIIa\,. Here, we are going to divide the interval $r_{5}$ into two parts: $(0,0.1]\cup(0.1,b/2]$.
\begin{itemize}
\item Part IIIa$_{1}$: for $r_{5}\in(0,0.1]$\,. Again we divide the interval of $\eta$ into $(1.16249,1.52721]\cup(1.52721,1.64631]$\,. 
\begin{itemize}
\item For $\eta\in(1.16249,1.52721]$\,. We are employing the Properties \ref{dfs}\,. The functions of the family $-\left(d_{\eta}f_{4_{\eta}}\right)(r_{5})$ are concave, enabling us to form the family $L_{41_{\eta}}(r_{5})$, comprising the tangent lines to each function of the family $-\left(d_{\eta}f_{4_{\eta}}\right) (r_{5})$\,. It is easy to see that the lines of the family $L_{41_{\eta}}(r_{5})$ are decreasing\,. Now, let us compare the family $\displaystyle\lim_{r_{5}\rightarrow 0}L_ {41_{\eta}}(r_{5})$ and $\left(d_{\eta}f_{1_{\eta}}+d_{\eta}f_{2_{\eta} }+d_{\eta}f_{3_{\eta}}\right)(0.1)$, which we can interpret as functions in terms of $\eta$\,. The function $\displaystyle\lim_{r_{5}\rightarrow 0}L_{41_{\eta}}(r_{5})$ is monotonically decreasing, while the function $\left(d_{\eta}f_{ 1_{\eta}}+d_{\eta}f_{2_{\eta} }+d_{\eta}f_{3_{\eta}}\right)(0.1)$ is increasing\,. Therefore, it follows that $\displaystyle\lim_{r_{5}\rightarrow 0}L_{41_{1.52721}}(r_{5})=3.81175<\left(d_{\eta}f_{1_{1.52721}}+d_{\eta}f_{2_{1.52721}}+d_{\eta}f_{3_{1.52721}}\right)(0.1)=10.0444$\,. 
\item For $\eta\in(1.52721,1.64631]$\,. We are going to divide the interval $r_{5}$ into $(0,0.0765], (0.0765,0.1]$.
\begin{itemize}
\item \underline{For $r_{5}\in(0,0.0765]$}\,. We are creating the family $L_{42_{\eta}}(r_{5})$, which consists of tangent lines to the functions of the family $-\left(d_{\eta}f_{4_{\eta}}\right)(r_{5})$ at the point $r_{5}=0.0765$\,. These lines are decreasing\,. We want to compare $\displaystyle{\lim _{r_{5}\rightarrow 0}L_{42_{\eta}}(r_{5})}$ and $\left(d_{\eta} f_{1_{\eta}}+d_{\eta}f_{2_{\eta}}+d_{\eta}f_{3_{\eta}}\right)(0.0765)$\,. We can view both families as functions in terms of $\eta$\,. The function $\displaystyle{\lim_{r_{5}\rightarrow 0}L_{42_{\eta}}(r_{5})}$ is decreasing, while the function $\left(d_{\eta} f_{1_{\eta}}+d_{\eta}f_{2_{\eta}}+d_{\eta}f_{3_{\eta}}\right)(0.0765)$ is increasing\,. So,
\begin{equation*}
\begin{split}
\lim_{r_{5}\rightarrow 0}L_{42_{1.5272}}(r_{5})&=3.75605\\
&<\left(d_{\eta}f_{1_{1.5272}}+d_{\eta}f_{2_{1.5272} }+d_{\eta}f_{3_{1.5272}}\right)(0.0765)\\&=11.7066.
\end{split}
\end{equation*}
So, $-\left(d_{\eta}f_{4_{\eta}}\right) (r_{5})<\left(d_{\eta}f_{1_{\eta}}+d_{\eta}f_{2_{\eta} }+d_{\eta}f_{3_{\eta}}\right)(r_{5})$.

\item \underline{For $r_{5}\in(0.0765,0.1]$}\,. The family of functions $\left(-d_{\eta}f_{4_{\eta}}\right)(r_{5})$ exhibits a decreasing behavior, as indicated by Properties \ref{dfs}\,. Consequently,  we compare $\left(-d_{\eta}f_{4_{\eta}}\right)(0.0765)$ and $\left(d_{\eta}f_{1_{\eta}}+d_{\eta}f_{2_{\eta}}+d_{\eta}f_{3_{\eta}}\right)(0.1)$ with respect to the parameter $\eta$\,. The former is decreasing, while the latter is increasing\,. So, 

\begin{equation*}
\begin{split}
-\left(d_{\eta}f_{4_{1.5272}}\right)(0.0765)&=2.82118\\
&<\left(d_{\eta}f_{1_{1.5272}}+d_{\eta}f_{2_{1.5272} }+d_{\eta}f_{3_{1.5272}}\right)(0.1)\\
&=10.0444.
\end{split}
\end{equation*}
Therefore, $-\left(d_{\eta}f_{4_{\eta}}\right)(r_{5})<\left(d_{\eta}f_{1_{\eta}}+d_{\eta}f_{2_{\eta} }+d_{\eta}f_{3_{\eta}}\right)(r_{5})$.
\end{itemize}

\end{itemize}
\item Part IIIa$_{2}$: for $r_{5}\in(0.1,b/2]$\,. We know that the functions of both families are decreasing (Properties \ref{dfs} and Proposition \ref{aux3})\,. We divide the interval of $r_{5}$ into $(0.1,0.35]\cup(0.35,b/2]$.
\begin{itemize}
\item For $r_{5}\in(0.1,0.35]$\,. We just compare $-\left(d_{\eta}f_{4_{\eta}}\right)(0.1)$ and $\left(d_{\eta}f_{1_{\eta}}+d_{\eta}f_{2_{\eta}}+d_{\eta}f_{3_{\eta}}\right )(0.35)$, which can be seen as functions in terms of $\eta$\,. An upper bound of the function $-\left(d_{\eta}f_{4_{\eta}}\right)(0.1)$ is $-\left(d_{\eta}f_{4_{1.28843}}\right)(0.1)=2.69234$\,. On the other hand, the function $\left(d_{\eta}f_{1_{\eta}}+d_{\eta}f_{2_{\eta} }+d_{\eta}f_{3_{\eta} }\right)(0.35)$ is increasing, so its minimum value is $\left(d_{\eta}f_{1_{1.16249}}+d_{\eta}f_{2_{1.16249}}+d_{\eta }f_{3_{1.16249}}\right)(0.35)=2.85468$\,. Then, $-\left(d_{\eta}f_{4_{\eta}}\right)(r_{5})<\left(d_{\eta}f_{1_{\eta}}+d_{\eta}f_{2_{\eta} }+d_{\eta}f_{3_{\eta}}\right)(r_{5})$.
\item For $r_{5}\in(0.35,b/2]$\,. We just compare $-\left(d_{\eta}f_{4_{\eta}}\right)(0.35)$ and $\left(d_{\eta}f_{1_{\eta}}+d_{\eta}f_{2_{\eta}}+d_{\eta}f_{3_{\eta}}\right )(b/2)$, which can be seen as functions that depend on $\eta$, and they are decreasing and increasing, respectively\,. Then, it is true that
\begin{equation*}
\begin{split}
-\left(d_{\eta}f_{4_{1.16249}}\right)(0.&35)=0.829044\\
&<\left(d_{\eta}f_{1_{1.16249}}+d_{\eta}f_{2_{1.16249}}+d_{\eta}f_{3_{1.16249}}\right)(b/2)=1.5501.
\end{split}
\end{equation*}
Therefore, $-\left(d_{\eta}f_{4_{\eta}}\right)(r_{5})<\left(d_{\eta}f_{1_{\eta}}+d_{\eta}f_{2_{\eta} }+d_{\eta}f_{3_{\eta}}\right)(r_{5})$.
\end{itemize}
\end{itemize}
\item Part IIIb: for $\eta\in(1.64631, 4.52952]$\,. The functions of the families $-\left(d_{\eta}f_{4_{\eta}}\right)(r_{5})$ and $\left(d_{\eta}f_{1_{\eta}}+d_{\eta}f_{2_{\eta} }+d_{\eta}f_{3_{\eta}}\right)(r_{5})$ are decreasing\,. We divide the interval for $r_{5}$ into $(0,0.03], (0.3,b/2]$.
\begin{itemize}
\item For $r_{5}\in(0,0.03]$\,. We compare $-\left(d_{\eta}f_{4_{\eta}}\right)(0)$ y $\left(d_{\eta}f_{1_{\eta}}+d_{\eta}f_{2_{\eta} }+d_{\eta}f_{3_{\eta}}\right)(0.3)$, which depend of $\eta$\,. Is easy to see that the maximum value of  $-\left(d_{\eta}f_{4_{\eta}}\right)(0)$ es $-\left(d_{\eta}f_{4_{1.82051}}\right)(0)=3.63409$, and the minimum value of $\left(d_{\eta}f_{1_{\eta}}+d_{\eta}f_{2_{\eta} }+d_{\eta}f_{3_{\eta}}\right)(0.3)$ is $\left(d_{\eta}f_{1_{0}}+d_{\eta}f_{2_{0} }+d_{\eta}f_{3_{0}}\right)(0.3)=3.69114$\,. Therefore, we have the inequality $-\left(d_{\eta}f_{4_{\eta}}\right)(r_{5})<\left(d_{\eta}f_{1_{\eta}}+d_{\eta}f_{2_{\eta} }+d_{\eta}f_{3_{\eta}}\right)(r_{5})$.
\item For $r_{5}\in(0.3,b/2]$\,. It is enough to compare $-\left(d_{\eta}f_{4_{\eta}}\right)(0.3)$ and $\left(d_{\eta}f_{1_{\eta}}+d_{\eta}f_{2_{\eta} }+d_{\eta}f_{3_{\eta}}\right )(b/2)$, which can be seen as functions that depend on $\eta$\,. It is easy to show that the function $-\left(d_{\eta}f_{4_{\eta}}\right)(0.3)$ is monotically decreasing, so its maximum value is $-\left(d_{\eta}f_{4_{1.64631}}\right )(0.3)=0.667847$\,. On the other hand, the function $\left(d_{\eta}f_{1_{\eta}}+d_{\eta}f_{2_{\eta} }+d_{\eta}f_{3_{\eta} }\right)(b/2)$ is monotonically increasing, so its minimum value is $\left(d_{\eta}f_{1_{1.64631}}+d_{\eta}f_{2_{1.64631} }+d_ {\eta}f_{3_{1.64631}}\right)(b/2)=1.60192.$
\end{itemize}
\item Part IIIc: for $\eta\in(4.52952, 7.60839]$\,. By  the Properties described in \ref{dfs}, the family of functions $-\left(d_{\eta}f_{4_{\eta}}\right)(r_{5})$ is strictly decreasing\,. Furthermore, due to the monotonicity of the functions of the family $\left(d_{\eta}f_{1_{\eta}}+d_{\eta}f_{2_{\eta} }+d_{\eta}f_{3_{\eta} }\right)(r_{5})$ (Proposition \ref{aux3}), we need to compare $-\left(d_{\eta}f_{4_{4.52952}}\right)(r_{5})$ and $\left(d_{\eta}f_{1_{\eta}}+d_{\eta}f_{2_{\eta} }+d_{\eta}f_{3_{\eta}}\right)( b/2)$\,. It is easy to see that the function $-\left(d_{\eta}f_{4_{4.52952}}\right)(r_{5})$ is decreasing, so its maximum value is $-\left(d_{\eta}f_{4_{4.52952}}\right)(0)=0.677958$\,. On the other hand, the family $\left(d_{\eta}f_{1_{\eta}}+d_{\eta}f_{2_{\eta}}+d_{\eta}f_{3_{\eta}}\right)(b/2)$ is decreasing, so its minimum value is $\left(d_{\eta}f_{1_{7.60839}}+d_{\eta}f_{2_{7.60839} }+d_{\eta}f_{3_{7.60839}}\right)(b/2)=1.67366.$
\end{itemize}

Therefore, $-\left(d_{\eta}f_{4_{\eta}}\right)(r_{5})<\left(d_{\eta}f_{1_{\eta}}+d_{\eta}f_{2_{\eta} }+d_{\eta}f_{3_{\eta}}\right)(r_{5})$.
\end{proof}

\begin{proposition}\label{auxextra1}
For $r_{5}\in (0,b/2]$ and $\eta\in(0,0.02]$, the functions of the family $\lambda_{31_{\eta}}(r_{5})$ are monotonic decreasing.
\end{proposition}

\begin{proof}
By to the Properties \ref{fs}, it remains to be shown that the proposition holds for $r_{5}\in(r_{5}^{*}(\eta),b/2]$ and $r_{5}\in(0,\hat{r_{5}}(\eta)]$.
\begin{itemize}
\item Part I: for $r_{5}\in(r_{5}^*(\eta),b/2]$\,. By the Properties \ref{fsp}, it is sufficient to compare the families of functions $-f_{2_{\eta}}'(r_{5})$ and $f_{3_{\eta}}'(r_{5})$\,. According to Properties \ref{fsp}, both families of functions are strictly increasing\,. Additionally, the functions of the family $-f_{2_{\eta}}'(r_{5})$ are decreasing monotone, while the functions of the family $f_{3_{\eta}}'(r_{5 })$ are increasing monotone\,. Therefore, it is sufficient to compare $-f_{2_{0}}'(b/2)$ and $f_{3_{0.02}}'(b/2)$ (see Fig.~\ref{fig15}). We find that $-f_{2_{0}}'(b/2)=1.68415$ is greater than $f_{3_{0.02}}'(b/2)=0.354561$\,. Therefore, $-f_{2_{\eta}}'(r_{5})>f_{3_{\eta}}'(r_{5})$.

\begin{figure}[ht]
	\centering
		\includegraphics[scale=0.8]{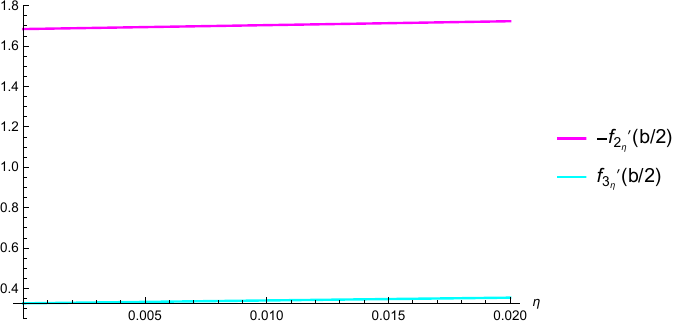}
	\caption{functions $-f_{2_{\eta}}'(b/2)$ and $f_{3_{\eta}}'(b/2)$.}
		\label{fig15}
\end{figure}

\item Part II: for $r_{5}\in(0,\hat{r_{5}}(\eta)]$\,. We will compare the functions of the families $-f_{2_{\eta}}'(r_{5})-f_{3_{\eta}}'(r_{5})$ and $f_{4_{\eta}}'(r_{5})$\,. We will divide the interval for $r_{5}$ into three parts: $(0,0.05]$, $(0.05,0.09]$, and $(0.09,\hat{r_{5}}(\eta)]$\,. In each interval, we will compare the lower bound of the family $-f_{2_{\eta}}'(r_{5})-f_{3_{\eta}}'(r_{5})$ with the upper bound of the family $f_{4_{\eta}}'(r_{5})$.
\begin{itemize}
\item For $r_{5}\in(0,0.05]$\,. 
We compare the families $-f_{2_{\eta}}'(0.05)-f_{3_{\eta}}'(0.05)$ and $f_{4_{\eta}}'(0)$, which we can view as functions depending on $\eta$\,. Both functions are increasing, so we construct the following piecewise functions (see Fig.~\ref{fig16}),

\begin{figure}[ht]
	\centering
		\includegraphics[scale=0.8]{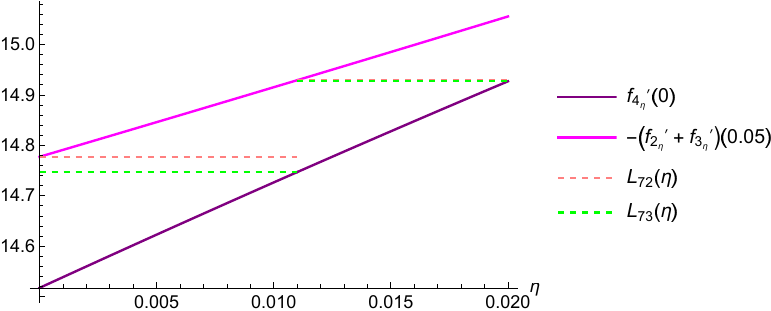}
	\caption{functions $-f_{4_{\eta}}'(0)$, $-(f_{2_{\eta}}'+f_{3_{\eta}}')(0.05)$, $L_{72}(\eta)$ and $L_{73}(\eta)$\,.}
		\label{fig16}
\end{figure}

\begin{equation*}
L_{72}(\eta):= \left\{ \begin{array}{lll}
            -f_{2_{0}}'(0.05)-f_{3_{0}}'(0.05)=14.7772\,, & 0\leq\eta\leq0.011\,,\\
 			-f_{2_{0.011}}'(0.05)-f_{3_{0.011}}'(0.05))=14.9292\,,& 0.011<\eta\leq0.02,
         \end{array}
   \right. 
\end{equation*}

\begin{equation*}
L_{73}(\eta):= \left\{ \begin{array}{lll}
            f_{4_{0.11}}'(0)=14.7469, & 0<\eta\leq0.011\,,\\
 			f_{4_{0.02}}'(0)=14.927,& 0.011<\eta\leq0.02\,,
         \end{array}
   \right. 
\end{equation*}

It follows that $L_{72}(\eta)>L_{73}(\eta)$, thus, $-f_{2_{\eta}}'(r_{5})-f_{3_{\eta} }'(r_{5})>f_{4_{\eta}}'(r_{5})$.

\item For $r_{5}\in(0.05,0.09]$\,. We compare the families $-f_{2_{\eta}}'(0.09)-f_{ 3_{\eta}}'(0.09)$ and $f_{4_{\eta}}'(0.05)$, which we can view as functions depending on $\eta$, (see Fig.~\ref{fig17})\,. The function $-f_{2_{\eta}}'( 0.09)-f_{3_{\eta}}'(0.09)$ is monotonically increasing, while $f_{4_{\eta}}'(0.05)$ is monotonically decreasing\,. We will compare the minimum and maximum values of each function, which are $-f_{2_{0}}'(0.09)-f_{3_{0}}'(0.09)=12.4475$ and $f_{4_{0}}'(0.05)=12.0126$\,. Thus, we have $-f_{2_{\eta}}'(r_{5})-f_{3_{\eta}}'(r_{5})>f_{4_{\eta}}'(r_{5})$.

\begin{figure}[ht]
	\centering
		\includegraphics[scale=0.8]{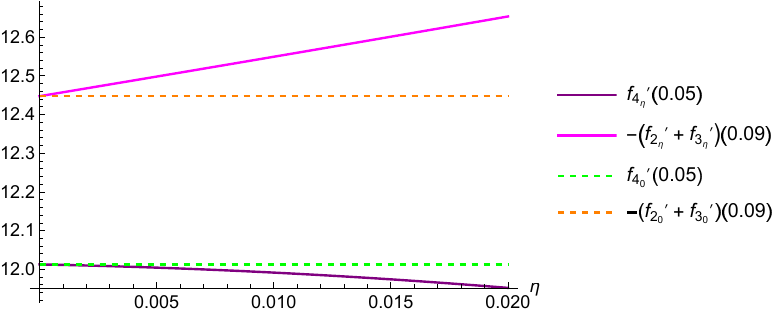}
	\caption{functions $f_{4_{\eta}}'(0.05)$ and $-(f_{2_{\eta}}'+f_{3_{\eta}}')(0.09)$.}
		\label{fig17}
\end{figure}

\item For $r_{5}\in(0.09,\hat{r_{5}}(\eta)]$\,. We compare the families  $-(f_{2_{\eta}}'+f_{3_{\eta}}')(\hat{r_{5}} (\eta))$ and $f_{4_{\eta}}'(0.09)$, which we can view as functions depending on $\eta$, (see Fig.~\ref{fig18})\,. The function $-f_{2_{\eta}}'(\hat{r_{5}}(\eta))-f_{3_{\eta}}'(\hat{r_{5}}(\eta) )$ is monotonically increasing, while $f_{4_{\eta}}'(0.09)$ is monotonically decreasing\,. We compare the minimum and maximum of each function, which are $-f_{2_{0}}'(\hat{r_{5}}(\eta))-f_{3_{0}}'(\hat{r_{5}}(\eta))=9.331$ and $f_{4_{0}}'(0.09)=8.25092$\,. Then, we conclude that $-f_{2_{\eta}}'(r_{5})-f_{3_{\eta}}'(r_{5})>f_{4_{\eta}}'(r_{5})$.

\begin{figure}[ht]
	\centering
		\includegraphics[scale=0.8]{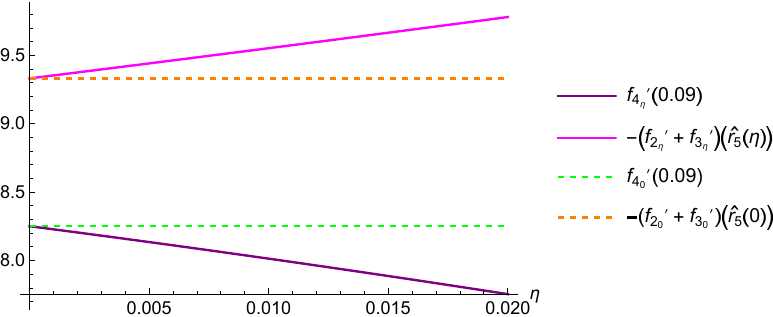}
	\caption{functions $f_{4_{\eta}}'(0.09)$ and $-(f_{2_{\eta}}'+f_{3_{\eta}}')(\hat{r_{5}}(\eta))$.}
		\label{fig18}
\end{figure}

Therefore, the functions of the family $\lambda_{31_{\eta}}(r_{5})$ are monotonically decreasing.
\end{itemize} 
\end{itemize}
\end{proof}

\subsection{Proofs of the regions $J_{n}$}\label{subsec:Jn}
\begin{proposition}\label{proposicionJ1}
For the region $J_{1}$, is true that $\lambda_{11_{\eta}}(r_{5})<\lambda_{31_{\eta}}(r_{5})$.
\end{proposition}

\begin{proof}
To present a more legible proof, we divide the domain of the parameter $\eta$ into $[0,0.02]\cup(0.02,\infty)$\,.

\begin{itemize}
\item \underline{Part I}: for $\eta\in[0.02,\infty)$, we consider $r_{5} \in (0,0.12874]$ and $r_{5}\in\left(0.12874,b/2\right]$.
\begin{itemize}
\item Part Ia: for $r_{5} \in (0,0.12874]$\,. By the Properties \ref{es}, we can verify that an upper bound of the family $\lambda_{11_\eta}(r_5)$ is given by $\displaystyle{\lim_{\eta\rightarrow \infty}e_{1_{\eta}}(0.12874)}+e_{2_{0.02}}(0)+e_{3_{\eta}}(r_{5_{\eta}})=2.05865+0.402987+1.0696=3.5312$\,. From the Proposition \ref{aux5} we have $\lambda_{31_{0.02}}(r_{5})\leq \lambda_{31_{\eta}}(r_{5})$ for all $\eta$ and $r_{5}$\,. We can verify  that $\lambda_{31_{0.02}}(r_{5})$ is a decreasing function and a lower bound is $\lambda_{31_{0.02}}(0.12874)=9.02703$\,. So $\lambda_{11_{\eta}}(r_{5})<\lambda_{31_{\eta}}(r_{5})$. (See Fig.~\ref{fig19a}).

\item Part Ib: for $r_{5}\in\left(0.12874,b/2\right]$\,. By the Properties \ref{es}, we can verify that an upper bound of the family $\lambda_{11_{\eta}}(r_{5})$ is $\displaystyle{\lim_{\eta\rightarrow \infty}}e_{1_{\eta}}(r_{5})+e_{2_{0.02}}(r_{5})+e_{3_{\eta}}(r_{5})$ which is an increasing function\,. Let $L_{8}(r_{5})$ be such this bound\,. We use the auxiliary function defined as $L_{9}(r_{5})$. (See Fig.~\ref{fig19b}).
\begin{equation*}
L_9(r_5)= \left\{ \begin{array}{lll}
            L_{8}(0.12874)=3.47274,  & 0.12874< r_5<0.5\,, \\
             \\  L_{8}(0.5)=2.75126, & 0.5\leq r_5<0.6\,,\\
             \\ L_{8}(0.6)=2.47871, & 0.6\leq r_5<b/2\,.
             \end{array}
   \right.
\end{equation*}

\begin{figure}[ht]
 \centering
  \subfloat[functions $\lambda_{11_{\eta}}(r_{5})$ and $\lambda_{31_{\eta}}(r_{5})$ for $\eta=7.2$]{
  \label{fig19a}
   \includegraphics[width=0.4\textwidth]{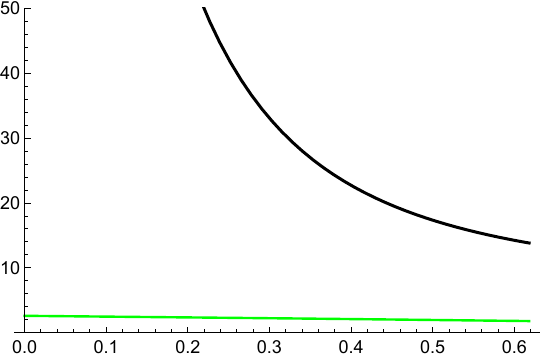}}
  \subfloat[functions $\lambda_{11_{\eta}}(r_{5})$ and $\lambda_{31_{\eta}}(r_{5})$ for $\eta=0.009$]{
   \label{fig19b}
    \ \ \ \ \ \includegraphics[width=0.5\textwidth]{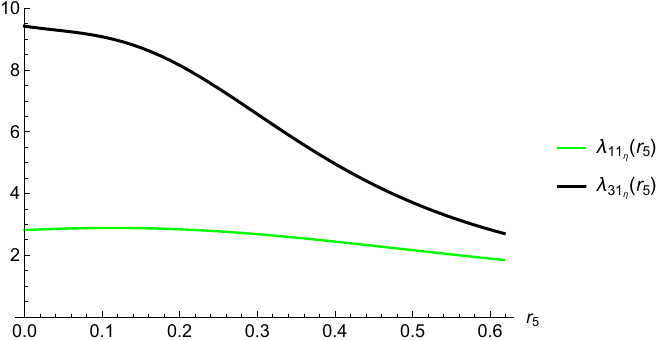}} 
   \caption{ }
\end{figure}

As the final item, $\lambda_{31_{0.02}}(r_{5})$ is a decreasing function\,. Therefore, we define the piecewise function:

\begin{equation*}
L_{10}(r_5)= \left\{ \begin{array}{lll}
            \lambda_{31_{0.02}}(0.5)=3.70904,  & 0.12874< r_5<0.5\,, \\
             \\  \lambda_{31_{0.02}}(0.6)=2.8317, & 0.5\leq r_5<0.6\,,\\
             \\ \lambda_{31_{0.02}}\left(b/2\right)=2.70691, & 0.6\leq r_5<b/2\,.
             \end{array}
   \right.
\end{equation*}

Therefore, $L_{9}(r_{5})>L_{10}(r_{5})$, and so $\lambda_{11_{\eta}}(r_{5})<\lambda_{31_{\eta}}(r_{5})$.
\end{itemize}
\item \underline{Part II}: for $\eta\in[0,0.02]$\,. Let us compare the function $e_{1_{0.02}}(r_{5})+e_{2_{0}}(r_{5})+e_{3_{\eta}}(r_{5_{\eta}})$, which serves as an upper bound for the family of functions $\lambda_{11_{\eta}}(r_{5})$ (as described in Properties \ref{es}), with the family of functions $\lambda_{31_{\eta}}(r_{5})$\,. We can calculate the maximum value of $e_{1_{0.02}}(r_{5})+e_{2_{0}}(r_{5})+e_{3_{\eta}}(r_{5_{\eta}})$ for $r_{5}=0.00985$\,. Thus, we can separate the domain of $r_{5}$ into three intervals: $(0,0.00985]$, $(0.00985,0.055]$, and $(0.55,b/2]$.

\begin{itemize}
\item Part IIa: for $r_{5}\in(0,0.00985]$\,. From the Proposition \ref{auxextra1} the functions of the family $\lambda_{31_{\eta}}(r_{5})$ are monotonically decreasing and a lower bound is $\lambda_{31_{0}}(0.00985)=9.25381$\,. On other hand, we calculate an upper bound of $e_{1_{0.02}}(r_{5})+e_{2_{0}}(r_{5})+e_{3_{\hat{\eta}(0)}}(r_{5})$ which is $e_{1_{0.02}}(0.00985)+e_{2_{0}}(0.00985)+e_{3_{\eta}}(r_{5})=2.88933$\,. Then $\lambda_{11_{\eta}}(r_{5})<\lambda_{31_{\eta}}(r_{5})$.

\item Part IIb: for $r_{5}\in(0.00985,0.55] $.We can verify the monotonicity of the family of functions $\lambda_{31_{\eta}}(r_{5})$, which is increasing (Proposition \ref{auxextra1})\,. We also calculate a lower bound of $\lambda_{31_{\eta}}(r_{5})$ which is $\lambda_{31_{0}}(0.55)=3.22774$\,. On other hand, we calculate the maximum value of  $e_{1_{0.02}}(r_{5})+e_{2_{0}}(r_{5})+e_{3_{\hat{\eta}(0)}}(r_{5})$, that is $e_{1_{0.02}}(0.55)+e_{2_{0}}(0.55)+e_{3_{\eta}}(r_{5_{\eta}})=1.86172$\,. Then, $\lambda_{11_{\eta}}(r_{5})<\lambda_{31_{\eta}}(r_{5})$.

\item Part IIc: for $r_{5}\in\left(0.55,b/2\right]$\,. We can establish the decreasing monotony behavior of the functions of the family $\lambda_{31_{\eta}}(r_{5})$ (Proposition \ref{auxextra1}), so the family of minimum values (parametrized by $\eta$) are $\lambda_{31_{\eta}}(b/2)$\,. This family is a strictly increasing family and a lower bounded is $\lambda_{31_{0}}\left(b/2\right)=2.70464$\,. As in the last items, the maximum value of the function $e_{1_{0.02}}(r_{5})+e_{2_{0}}(r_{5})+e_{3_{\hat{\eta}(0)}}(r_{5})$, corresponds to \mbox{$e_{1_{0.02}}(0.55)+e_{2_{0}}(0.55)+e_{3_{\eta}}(r_{5_{\eta}})=1.86172$}\,\,. Therefore, $\lambda_{11_{\eta}}(r_{5})<\lambda_{31_{\eta}}(r_{5})$.
\end{itemize}
\end{itemize}
\end{proof}

\begin{proposition}\label{proposicionJ2}
For the region $J_{2}$, is true that $\lambda_{41_{\eta}}(r_{5})<\lambda_{31_{\eta}}(r_{5})$.
\end{proposition}

\begin{proof}
For monotonicity, we divide the interval of $\eta$ into $[0,0.02]$ and $(0.02,\infty)$.

\begin{itemize}
\item \underline{Part I}: for $\eta\in[0,0.02]$\,. We compare the function $\lambda_{41_{0}}(r_{5})$, which is an upper bound of the family $\lambda_{41_{\eta}}(r_{5})$, with the family $\lambda_{31_{\eta}}(r_{5})$, which functions are monotonically decreasing. We can verify that the family of values $\lambda_{31_{\eta}}(1)$ is strictly increasing\,, therefore, $\lambda_{31_{0}}(1)=1.36009$ is a lower bound of this family\,. On the other hand, we can verify that the function $\lambda_{41_{0}}(r_{5})$ is monotonically decreasing and its maximum value is $\lambda_{41_{0}}\left(b/2\right)=1.32438$\,. Hence, $\lambda_{31_{\eta}}(r_{5})>\lambda_{41_{\eta}}(r_{5})$\,. See Fig.~\ref{fig20a}.

\item \underline{Part II}: for $\eta\in(0.02,\infty)$\,. We can verify that a lower bound for the family $\lambda_{31_{\eta}}(r_{5})$  is $\lambda_{31_{0.02}}(r_{5})$ and the maximal function of the family $\lambda_{41_{\eta}}(r_{5})$ is $\lambda_{41_{0.02}}(r_{5})$\,. Both functions are monotonically decreasing, so the minimum value of $\lambda_{31_{0.02}}(r_{5})$ is $\lambda_{31_{0.02}}(1)=1.37246$ and the maximum value of $\lambda_{41_{0.02}}(r_{5})$ is $\lambda_{41_{0.02}}\left(b/2\right)=1.30144$\,. Thus $\lambda_{31_{0.02}}(r_{5})>\lambda_{41_{0.02}}(r_{5})$\,. See Fig.~\ref{fig20b}.

\end{itemize}
Therefore $\lambda_{31_{\eta}}(r_{5})>\lambda_{41_{\eta}}(r_{5})$.

\begin{figure}[ht]
 \centering
  \subfloat[functions $\lambda_{31_{\eta}}(r_{5})$ and $\lambda_{41_{\eta}}(r_{5})$ for $\eta=0.0102$]{
  \label{fig20a}
   \includegraphics[width=0.4\textwidth]{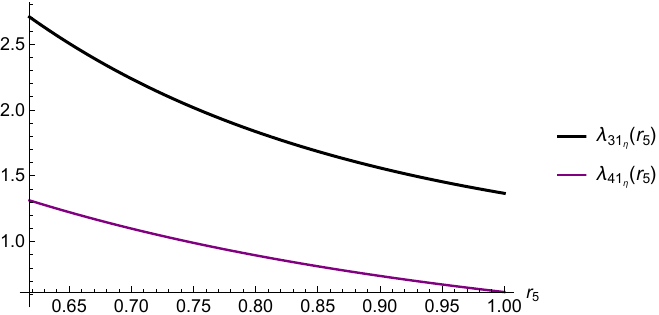}}
  \subfloat[functions $\lambda_{31_{\eta}}(r_{5})$ and $\lambda_{41_{\eta}}(r_{5})$ for $\eta=8.3$]{
   \label{fig20b}
    \ \ \ \ \ \includegraphics[width=0.4\textwidth]{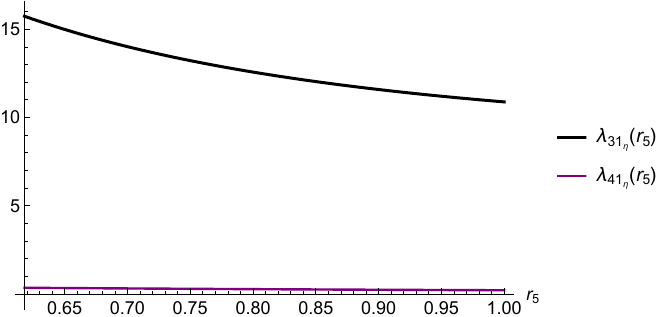}} 
   \caption{ }
\end{figure}
\end{proof}

For the regions $J_{3}$, $J_{4}$ and $J_{5}$ we define the following family of functions:

\begin{equation*}
\lambda_{ik_{\zeta}}(r_{3})\colonequals\lambda_{ik}\left(r_{3},\frac{b}{2+\zeta}\right)\,, \zeta\in [0,nd(\zeta)),
\end{equation*}
here $nd(\zeta)\colonequals b/2\left(\frac{4+\zeta}{2+\zeta}\right)$\,. Also we define $nz\colonequals\frac{4(b-1)}{2-b}$.

\begin{proposition}\label{proposicionJ3}
For the region $J_{3}$, is true that $\lambda_{11_{\zeta}}(r_{3})<\lambda_{52_{\eta}}(r_{3})$, for $\zeta\in[0,nz)$.
\end{proposition}

\begin{proof}
The proof involves establishing a bound for the families of functions \(\lambda_{11_{\zeta}}(r_{3})\) and \(\lambda_{52_{\zeta}}(r_{3})\). We can determine which families of functions are strictly decreasing, and it is sufficient to demonstrate that \(\lambda_{11_{nz}}(r_{3}) < \lambda_{52_{0}}(r_{3})\). To do this, we divide the domain of \(r_{3}\) into two intervals: \(\left(b/2, 0.68\right]\) and \((0.68, 1]\).

\begin{itemize}
\item \underline{Part I}: for $r_{3}\in\left(b/2,0.68\right]$\,. We defined the functions using the following piecewise functions: \( L_{18}(r_{3}) \) and \( L_{19}(r_{3}) \). See Fig.~\ref{fig21a}.

\begin{equation*}
L_{18}(r_3)= \left\{ \begin{array}{lll}
          \lambda_{52_{0}}\left(b/2\right)=2.57116,  &b/2< r_3 \leq 0.66\,, \\
             \\    \lambda_{52_{0}}\left(0.66\right)=2.58854, &0.66<r_3\leq0.674\,, \\
						 \\    \lambda_{52_{0}}\left(0.674\right)=2.60018, &0.674<r_3\leq 0.68\,, 
             \end{array}
   \right.
\end{equation*}

\begin{equation*}
L_{19}(r_3)= \left\{ \begin{array}{lll}
          \lambda_{11_{nz}}\left(0.66\right)=2.56495,  &b/2< r_3 \leq 0.66\,, \\
             \\    \lambda_{11_{nz}}\left(0.674\right)=2.58741, &0.66<r_3\leq 0.674\,, \\
						 \\    \lambda_{11_{nz}}\left(0.68\right)=2.596902, &0.674<r_3 \leq 0.68\,.
             \end{array}
   \right.
\end{equation*}

Clearly, $L_{19}(r_{3})<L_{18}(r_{3})$, then $\lambda_{11_{nz}}(r_{3})<\lambda_{52_{0}}(r_{3})$\,.

\item \underline{Part II}: for $r_{3}\in(0.68,1]$\,. The function  $\lambda_{52_{0}}(r_{3})$ is convex. Consider the linear function $L_{20}(r_{3}) = \frac{2.6 - \lambda_{11_{nz}}(0.72)}{0.68 - 0.72}(r_{3} - 0.68) + 2.601$. It holds that  $L_{20}(r_{3}) > \lambda_{11_{nz}}(r_{3})$. Additionally, we can analyze the tangent lines $L_{21}(r_{3})$ and $L_{24}(r_{3})$ to the function $\lambda_{52_{0}}(r_{3}) - L_{20}(r_{3})$ at the points $r_{3} = 0.68$, $r_{3} = 0.7$, $r_{3} = 0.72$, and $r_{3} = 1$. We can verify that the intersections of these lines occur at positive points, which leads us to conclude that $\lambda_{52_{0}}(r_{3}) > L_{20}(r_{3})$.

Finally, we construct the tangent lines to $L_{20}(r_{3})-\lambda_{11_{nz}}(r_{3})$ for $r_{3}=0.68$, $r_{3}=0.73$ and $r_{3}=1$\,. Let $L_{25}(r_{3})$ to $L_{27}(r_{3})$ be such these tangents lines\,. We can verify that the intersection of the lines is in positive points, then $L_{20}(r_{3})>\lambda_{11_{nz}}(r_{3})$\,. Then, $\lambda_{11_{nz}}(r_{3})<\lambda_{52_{0}}(r_{3})$. See Fig.~\ref{fig21b}.
\end{itemize}

Therefore,  $\lambda_{11_{\zeta}}(r_{3})<\lambda_{52_{\zeta}}(r_{3})$. 
\begin{figure}[ht]
 \centering
  \subfloat[Part I]{
  \label{fig21a}
   \includegraphics[width=0.45\textwidth]{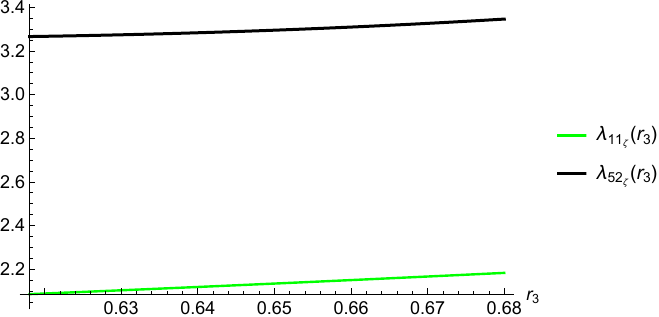}}
  \subfloat[Part II]{
   \label{fig21b}
    \ \ \ \ \ \includegraphics[width=0.45\textwidth]{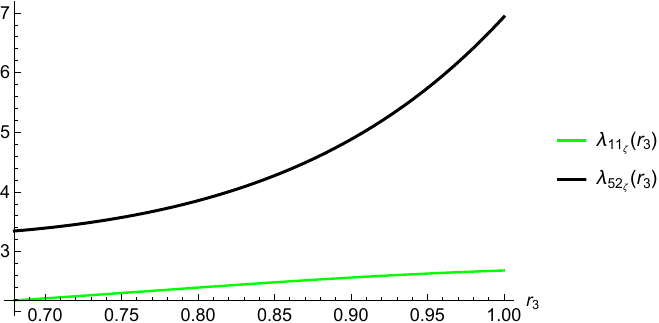}} 
   \caption{functions $\lambda_{11_{\zeta}}(r_{3})$ and $\lambda_{52_{\zeta}}(r_{3})$ for $\eta=0.328794$ }
\end{figure}
\end{proof}

\begin{proposition}\label{proposicionJ4}
For the region $J_{4}$, is true that $\lambda_{11_{\zeta}}(r_{3})<\lambda_{52_{\zeta}}(r_{3})$, $\zeta\in[nz,nd(\zeta))$.
\end{proposition}

\begin{proof}
An upper bound for $\lambda_{11_{\zeta}}(r_{3})$ is $2.95524$. The function $\lambda_{52_{nz}}(r_{3})$ is increasing and its minimal value is $\lambda_{52_{nz}}\left(b/2\right)=5.76142$. Therefore, $\lambda_{11_{\zeta}}(r_{3})<\lambda_{52_{\zeta}}(r_{3})$\,. See Fig.~\ref{fig22}.

\begin{figure}[ht]
	\centering
		\includegraphics[scale=0.7]{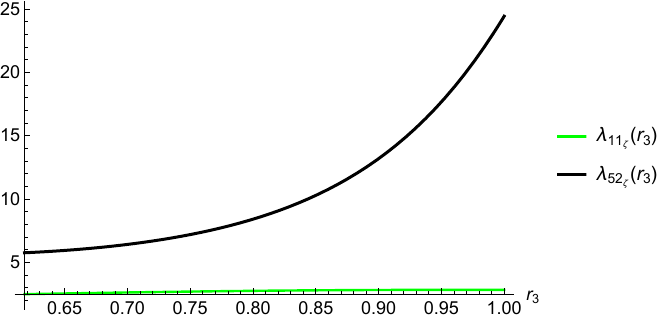}
	\caption{functions $\lambda_{11_{\zeta}}(r_{3})$ and $\lambda_{52_{\zeta}}(r_{3})$ for $\zeta=1.23607$\,.}
		\label{fig22}
\end{figure}
\end{proof}

\begin{proposition}\label{proposicionJ5}
For the region $J_{5}$, is true that $\lambda_{11_{\zeta}}(r_{3})<\lambda_{52_{\zeta}}(r_{3})$, $\zeta\in[0,nd(\zeta))$.
\end{proposition}

\begin{proof}
We realize that the family of functions $\lambda_{52_{\zeta}}(r_{3})$ is increasing, then the function $\lambda_{52_{0}}(r_{3})$ is a lower bound function and this is monotonically increasing with a minimum value $\lambda_{52_{0}}(1)=4.4042$\,. On the other hand, we verify that $\lambda_{11_{\zeta}}(r_{3})<2.90749$. Therefore, $\lambda_{11_{\zeta}}(r_{3})<\lambda_{52_{\zeta}}(r_{3})$. See Fig.~\ref{fig23}.

\begin{figure}[ht]
	\centering
		\includegraphics[scale=0.8]{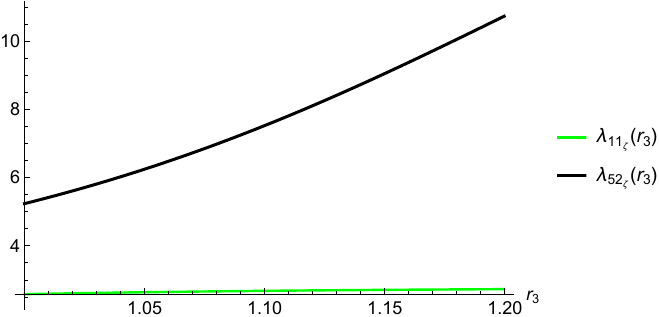}
	\caption{functions $\lambda_{11_{\zeta}}(r_{3})$ and $\lambda_{52_{\zeta}}(r_{3})$ for $\zeta=1.25$\,.}
		\label{fig23}
\end{figure}
\end{proof}

For the region $J_{6}$, we have defined the following family of functions:
\begin{equation*}
\lambda_{ik_{\mu}}(r_{5})\colonequals\lambda_{ik}\left(b+\mu, r_{5}\right)\,,  \quad \mu\in\left[0,2/b-b\right)\,.
\end{equation*}

\begin{proposition}\label{proposicionJ6}
For the region $J_{6}$, is true that $\lambda_{11_{\mu}}(r_{5})>\lambda_{31_{\mu}}(r_{5})$.
\end{proposition}

\begin{proof}
We can verify that the family of functions $\lambda_{11_{\mu}}(r_{5})$ are strictly increasing and the family of functions $\lambda_{31_{\mu}}(r_{5})$ are strictly decreasing. So, it is enough to compare the functions $\lambda_{11_{0}}(r_{5})$ and $\lambda_{31_{0}}(r_{5})$. An easy computation shows that the minimal value of the function $\lambda_{11_{0}}(r_{5})$ is $\lambda_{11_{0}}(1)=1.84995$. On the other hand, an upper bound for the function $\lambda_{31_{0}}(r_{5})$ is 
$1.60778$. So, $\lambda_{11_{0}}(r_{5})>\lambda_{31_{0}}(r_{5})$.
Therefore, $\lambda_{11_{\mu}}(r_{5})>\lambda_{31_{\mu}}(r_{5})$. See Fig.~\ref{fig24}.

\begin{figure}[ht]
	\centering
		\includegraphics[scale=0.8]{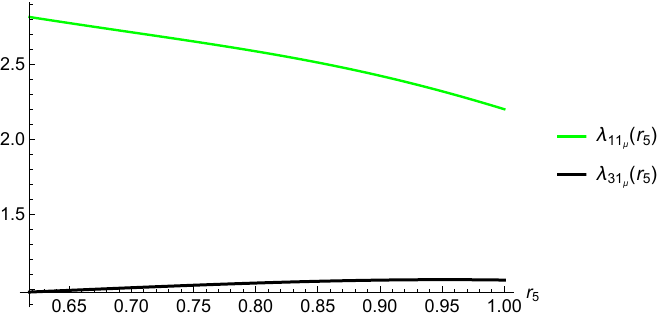}
	\caption{functions $\lambda_{11_{\mu}}(r_{3})$ and $\lambda_{31_{\mu}}(r_{3})$ for $\mu=0.1543145$\,.}
		\label{fig24}
\end{figure}
\end{proof}

For the regions $J_{7}$ and $J_{8}$, we defined the following family of functions 

\begin{equation*}
\lambda_{ik_{\iota}}(r_{3})\colonequals\lambda_{ik}\left(r_{3},1-\iota\right)\,,  \quad \iota\in\left(0,1-b/2\right)\,.
\end{equation*}

\begin{proposition}\label{proposicionJ7}
For the region $J_{7}$, is true that $\lambda_{52_{\iota}}(r_{3})>\lambda_{42_{\iota}}(r_{3})$.
\end{proposition}

\begin{proof}
We can write the families of functions $\lambda_{52_{\iota}}(r_{3})$ and $\lambda_{42_{\iota}}(r_{3})$ as follow $\lambda_{52_{\iota}}(r_{3})=H_{1_{\iota}}(r_{3})+H_{2_{\iota}}(r_{3})+H_{3_{\iota}}(r_{3})+H_{4_{\iota}}(r_{3})$ and $\lambda_{42_{\iota}}(r_{3})=G_{1_{\iota}}(r_{3})+G_{2_{\iota}}(r_{3})+G_{3_{\iota}}(r_{3})+G_{4_{\iota}}(r_{3})$. The families of function $H_{1_{\iota}}(r_{3})-G_{1_{\iota}}(r_{3})$, $H_{3_{\iota}}(r_{3})-G_{3_{\iota}}(r_{3})$ and $H_{2_{\iota}}(r_{3})+H_{4_{\iota}}(r_{3})-G_{2_{\iota}}(r_{3})-G_{4_{\iota}}(r_{3})$ are strictly positive.  Therefore $\lambda_{52_{\iota}}(r_{3})>\lambda_{42_{\iota}}(r_{3})$. See Fig.~\ref{fig25}.
\begin{figure}[ht]
	\centering
		\includegraphics[scale=0.8]{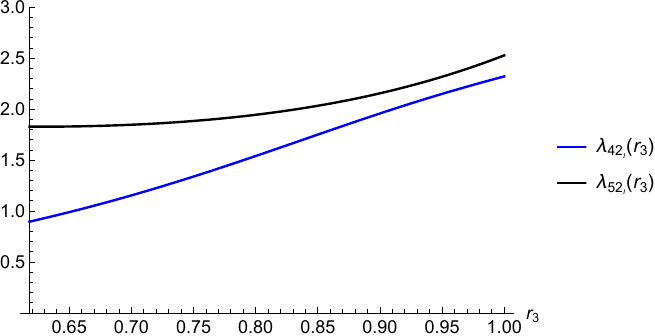}
	\caption{functions $\lambda_{52_{\iota}}(r_{3})$ and $\lambda_{42_{\iota}}(r_{3})$ for $\iota=0.232235$\,.}
		\label{fig25}
\end{figure}
\end{proof}

\begin{proposition}\label{proposicionJ8}
For the region $J_{8}$, is true that $\lambda_{22_{\iota}}(r_{3})>\lambda_{32_{\iota}}(r_{3})$.
\end{proposition}

\begin{proof}
We can write $\lambda_{22_\iota}(r_3)$ as follow $ \lambda_{22_\iota}(r_3)=K_{1_\iota}(r_{3})+...+K_{4_\iota}(r_{3})$ and $\lambda_{32_\iota}(r_3)$ as $\lambda_{32_\iota}(r_3)=z_{1_{\iota}}(r_{3})+...+z_{4_{\iota}}(r_{3})$.
We define the following families of functions:
$L_{60_{\iota}}(r_{3}):=K_{1_\iota}(r_{3})+K_{2_\iota}(r_{3})-z_{1_{\iota}}(r_{3})-z_{4_{\iota}}(r_{3})$ and 
$L_{61_{\iota}}(r_{3}):=K_{3_\iota}(r_{3})+K_{4_\iota}(r_{3})-z_{2_{\iota}}(r_{3})-z_{3_{\iota}}(r_{3})$. We have that 
the family of function $L_{60_{\iota}}(r_{3})+L_{61_{\iota}}(r_{3})$ is strictly positive. See Fig.~\ref{fig26}. 

\begin{figure}[ht]
	\centering
		\includegraphics[scale=0.8]{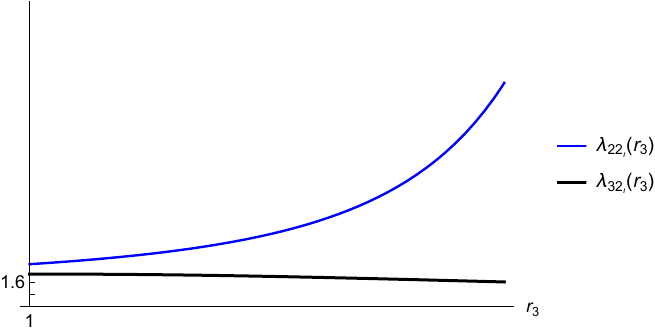}
	\caption{functions $\lambda_{22_{\iota}}(r_{3})$ and $\lambda_{32_{\iota}}(r_{3})$ for $\iota=0.237965$\,.}
		\label{fig26}
\end{figure}
\end{proof}

For the regions $J_{9}$ to $J_{15}$, we defined the follow family of functions 

\begin{equation*}
\lambda_{ik_{\xi}}(r_{3})\colonequals\lambda_{ik}\left(r_{3},1+\xi\right)\,.\end{equation*}

\begin{proposition}\label{proposicionJ9}
For the region $J_{9}$, is true that $\lambda_{11_{\xi}}(r_{3})>\lambda_{41_{\xi}}(r_{3})$, $\xi\in(0,\infty)$.
\end{proposition}

\begin{proof}
We divide the domain of the parameter $\xi$ as follows $\xi\in[0,0.99412)\cup [0.99412,\infty)$.

\begin{itemize}
\item \underline{Part I}: for $\xi\in[0,0.99412)$\,. We define a $\xi$-parameterized family of piecewise functions, each consisting of two tangent lines. The first line, $L_{47_{\xi}}(r_3)$, is the tangent line at $r_3 = 0$ with respect to $\lambda_{11_{\xi}}(r_3)$, and the second line, $L_{48_{\xi}}(r_3)$, is the tangent line at $r_3 = 1$. In this way, the piecewise function is the following:

\begin{equation*}
\bar{L}_{1_{\xi}}(r_{3})= \left\{ \begin{array}{lll}
          L_{47_{\xi}}(r_{3}),  &0< r_3 \leq r^{*}_{\xi}\,, \\
             \\    L_{48_{\xi}}(r_{3}), &r^{*}_{\xi}<r_3\leq 1\,,
             \end{array}
   \right.
\end{equation*}
where $r^{*}_{\xi}=\{r_{3}\colon L_{47_{\xi}}(r_{3})=L_{48_{\xi}}(r_{3})\}$. Let $L_{49_\xi}(r_3)$ and $L_{50_\xi}(r_3)$ represent the family of lines connecting the endpoints of functions in the family $\lambda_{41_\xi}(r_3)$ over the intervals $(0, 0.6)$ and $(0.6,1)$, respectively. We have that $L_{48_{\xi}}(r_{3})>L_{50_{\xi}}(r_{3})$ for all $r_{3}\in(0.6,1)$, and $L_{47_{\xi}}(r^{*}_{\xi})>L_{49_\xi}(r_3)$. See Fig.~\ref{fig27}.

\begin{figure}[ht]
	\centering
		\includegraphics[scale=0.8]{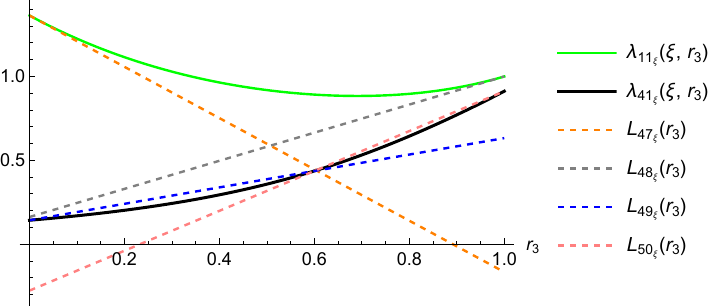}
		\caption{functions $L_{47_\xi}(r_{3})$, $L_{48_\xi}(r_{3})$, $L_{49_\xi}(r_{3})$ and $L_{50_\xi}(r_{3})$ for $\xi=0.189$.}
		\label{fig27}
\end{figure}

\item \underline{Part II}: for $\xi\in[0.99412, \infty)$\,. Since the family $\lambda_{11_\xi}(r_3)$ is strictly decreasing, we will analyze the function obtained by calculating $\lim_{\xi\rightarrow \infty} \lambda_{11_\xi}(r_3)$, which is a monotonically decreasing function. On the other hand, the family $\lambda_{41_\xi}(r_3)$ is decreasing, and each function within this family is increasing. Thus, we can denote $\lambda_{41_{0.99412}}(1) = 1.39635$ as an upper bound for $\lambda_{41_\xi}(r_3)$.  We have the inequality $\lambda_{11_\xi}(r_3) > \lim_{\xi\rightarrow \infty} \lambda_{11_\xi}(r_3) >\lambda_{41_{0.99412}}(1)\geq \lambda_{41_\xi}(r_3)$.
\end{itemize}
\end{proof}

\begin{proposition}\label{proposicionJ10}
For the region $J_{10}$, is true that $\lambda_{11_{\xi}}(r_{3})>\lambda_{41_{\xi}}(r_{3})$, $\xi\in[b,\infty)$.
\end{proposition}

\begin{proof}
The functions in the family $ \lambda_{11_\xi}(r_{3}) $ are convex. We define the $\xi$-parameterized family of piecewise functions, $\bar{L}_{2_{\xi}}(r_{3})$. The first family, $ L_{91_\xi}(r_3) $, represents the family of tangent lines at $ r_{3}= 1 $. The second family, $ L_{92_\xi}(r_3) $, corresponds to the family of tangent lines at $ r_{3} = 1.3 $. And the third family, $ L_{93_\xi}(r_3) $ is the family of tangent lines at $ r_{3}= 2/b$. Namely,

\begin{equation*}
\bar{L}_{2_{\xi}}(r_{3})= \left\{ \begin{array}{lll}
          L_{91_{\xi}}(r_{3}),  &0< r_3 \leq r^{*}_{1_{\xi}}\,, \\
             \\    L_{92_{\xi}}(r_{3}), &r^{*}_{1_{\xi}}<r_3\leq r^{*}_{2_{\xi}}\,, \\
             \\    L_{93_{\xi}}(r_{3}), &r^{*}_{2_{\xi}}<r_3\leq 1\,,
             \end{array}
   \right.
\end{equation*}

where $r^{*}_{1_{\xi}}=\{r_{3}\colon L_{91_{\xi}}(r_{3})=L_{92_{\xi}}(r_{3})\}$ and  $r^{*}_{2_{\xi}}=\{r_{3}\colon L_{92_{\xi}}(r_{3})=L_{93_{\xi}}(r_{3})\}$. See Fig.~\ref{fig28}.

\begin{figure}[ht]
	\centering
		\includegraphics[scale=0.9]{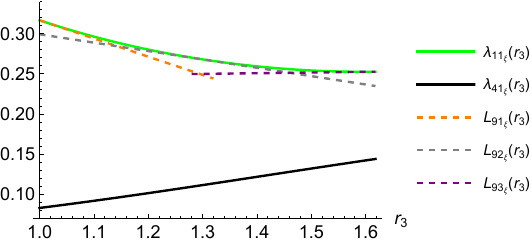}
		\caption{functions $\lambda_{11_{\xi}}(r_{3}), \lambda_{41_{\xi}}(r_{3}), L_{91_{\xi}}(r_{3}), L_{92_{\xi}}(r_{3})$ y $L_{93_{\xi}}(r_{3})$ for $\xi=2.09$.}
		\label{fig28}
	\end{figure}	
	
The minimum values of $ r^{*}_{1_{\xi}} $ and $ r^{*}_{2_{\xi}} $ are $ 1.13653 $ and $ 1.44536 $, respectively. On the other hand, since the family $\lambda_{41_{\xi}}(r_{3})$ is strictly decreasing and the functions of the family are increasing, the maximum value of the family   is given by $ \lambda_{41_{b}}\left(2/b \right) = 0.360157 $. Therefore, it follows that  $ \lambda_{11_{\xi}}(r_3) > \lambda_{41_{\xi}}(r_3). $
\end{proof}

\begin{proposition}\label{proposicionJ111213}
For the regions $J_{11}$, $J_{12}$ y $J_{13}$ is true that $\lambda_{11_{\xi}}(r_{3})>\lambda_{51_{\xi}}(r_{3})$ considering $\xi\in(0,2.036)$, $\xi\in(0,1.05)$ and $\xi\in(0.4,1.05)$, respectively.
\end{proposition}

\begin{proof}
We will show the proof by regions.
\begin{itemize}
\item For the region $J_{11}$\,: we will compare the family formed by the minimum values of the functions in the family $ \lambda_{11_{\xi}}(r_3) $ with the family formed by the maximum values of the functions in the family $ \lambda_{51_{\xi}}(r_3) $. Specifically, we will examine the curves of values $ \xi $-parameterized: $ \lambda_{11_{\xi}}\left(2/b\right) $ and $ \lambda_{51_{\xi}}\left(2/b\right) $.
Both functions are monotonically decreasing. Hence, we will construct the piecewise functions $L_{121}(\xi)$ and  $L_{122}(\xi)$ as shown in Fig.~\ref{fig29a}.

So, it is true that $\lambda_{11_{\xi}}\left(2/b\right)>L_{122}(\xi)>L_{121}(\xi)>\lambda_{51_{\xi}}\left(2/b\right)$. Therefore, $\lambda_{11_{\xi}}(r_{3})>\lambda_{51_{\xi}}(r_{3})$.

\item For the region $J_{12}$\,: the family of functions $ \lambda_{11_{\xi}}(r_3) $ is strictly decreasing, while the functions within this family are monotonically increasing. Additionally, the family $ \lambda_{51_{\xi}}(1.3) $ is also strictly decreasing and serves as an upper bound for the family of functions $ \lambda_{51_{\xi}}(r_3) $. Therefore, it is sufficient to compare the decreasing families $ \lambda_{11_{\xi}}(1.3) $ and $ \lambda_{51_{\xi}}(1.3) $. We can construct the piecewise functions $L_{130}(\xi)$ and $L_{131}(\xi)$ as show in Fig.~\ref{fig29b}. We have that, $\lambda_{11_{\xi}}(1.3)>L_{131}(\xi)>L_{130}(\xi)>\lambda_{51_{\xi}}(1.3)$. Therefore, $\lambda_{11_{\xi}}(r_{3})>\lambda_{51_{\xi}}(r_{3})$.

\begin{figure}[ht]
 \centering
  \subfloat[Functions $\lambda_{11_{\xi}}\left(2/b\right)$, $\lambda_{51_{\xi}}\left(2/b\right)$, $L_{121}(\xi)$ and $L_{122}(\xi)$ . 
  ]
  {
  \label{fig29a}
   \includegraphics[width=0.45\textwidth]{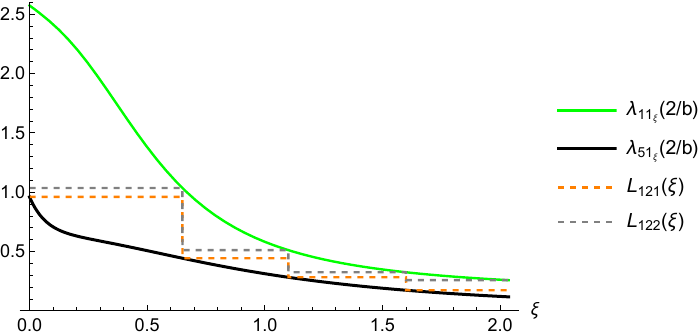}}
  \subfloat[$\lambda_{11_{\xi}}\left(1.3\right)$, $\lambda_{51_{\xi}}\left(1.3\right)$, $L_{130}(\xi)$ and $L_{131}(\xi)$
  ]
  {
   \label{fig29b}
    \ \ \includegraphics[width=0.45\textwidth]{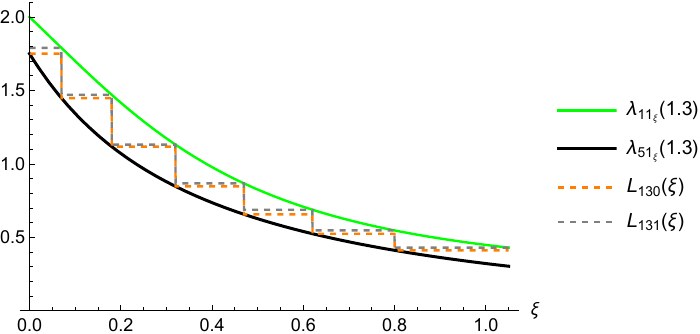}} 
   \caption{}
\end{figure}

\item For the region $J_{13}$: the family of functions $ \lambda_{51_{\xi}}(r_3) $ is strictly decreasing, with monotonically increasing functions. Thus, we consider $ \lambda_{51_{\xi}}(1.3) $, a monotonically decreasing function of $ \xi $. On the other hand, the functions in the family $ \lambda_{11_{\xi}}(r_3) $ are convex. Therefore, we can consider the family of points, $ P_6(\xi) $, formed by the intersection of the families of tangent lines to the functions in the family $ \lambda_{11_{\xi}}(r_3) $ at the endpoints of the interval for $ r_3 $. This function is decreasing; so, 
we compare $ P_6(\xi) $ and $ \lambda_{51_{\xi}}(1.3) $. Both are decreasing, with $ P_6(1.05) = 1.15788 $ and $ \lambda_{51_{0.4}}(1.3) = 0.736328 $. See FIGs.~\ref{fig30a} and~\ref{fig30b}. Therefore $\lambda_{11_{\xi}}(r_{3})>\lambda_{51_{\xi}}(r_{3})$.

\begin{figure}[ht]
 \centering
  \subfloat[Functions $\lambda_{11_{\xi}}\left(r_{3}\right)$, $\lambda_{51_{\xi}}\left(1.3\right)$ and tangent lines of $\lambda_{11_{\xi}}\left(r_{3}\right)$. 
  ]
  {
  \label{fig30a}
   \includegraphics[width=0.45\textwidth]{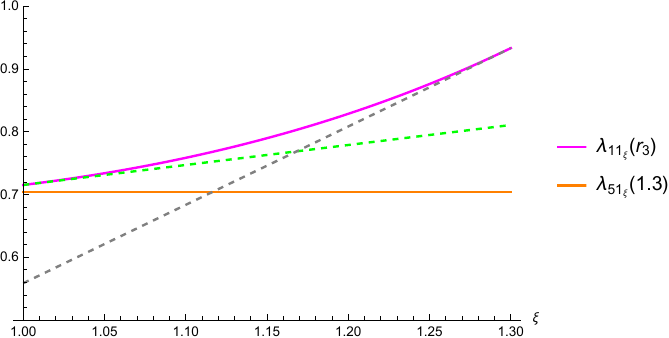}}
  \subfloat[Functions $P_{6}(\xi)$ and $\lambda_{51_{\xi}}\left(1.3\right)$.
  ]
  {
   \label{fig30b}
    \ \ \includegraphics[width=0.45\textwidth]{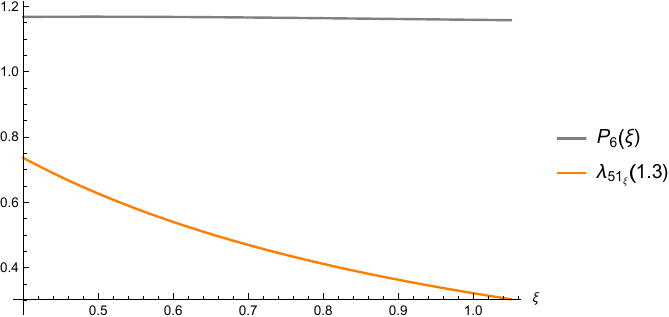}} 
   \caption{}
\end{figure}
\end{itemize}	
\end{proof}	

\begin{proposition}\label{proposicionJ14}
For the region $J_{14}$, is true that $\lambda_{11_{\xi}}(r_{3})>\lambda_{31_{\xi}}(r_{3})$, $\xi\in[1.05,b)$.
\end{proposition}

\begin{proof}
The families of functions $ \lambda_{11_{\xi}}(r_3) $ and $ \lambda_{31_{\xi}}(r_3) $ are strictly decreasing. So, 
it suffices to compare the minimal function in the family $ \lambda_{11_{\xi}}(r_3) $ with the maximal function in the family $ \lambda_{31_{\xi}}(r_3) $, which are $ \lambda_{11_b}(r_3) $ and $ \lambda_{31_{1.05}}(r_3) $, respectively. We will divide the interval of $ r_3 $ into $ (1, 1.195585] \cup (1.195585, 1.44496] \cup \left( 1.44496, 2/b \right) $.
For the subinterval $ (1, 1.195585] $, we will construct the piecewise functions $ L_{94}(r_3) $ and $ L_{95}(r_3) $ shown in Fig.~\ref{fig31a}. It is true that, $ L_{95}(r_3) > L_{94}(r_3) $, so $ \lambda_{11_b}(r_3) > \lambda_{31_{1.05}}(r_3) $.	
For the subinterval $(1.195585, 1.44496]$, we construct the line $L_{96}(r_3)$, which is tangent to the function $\lambda_{11_b}(r_3)$ at $r_3 = 1.44496$. Additionally, we define the line $L_{97}(r_3)$, which connects the endpoints of the function $\lambda_{31_{1.05}}(r_3)$. We can establish the inequality $
\lambda_{11_b}(r_3) > L_{96}(r_3) > L_{97}(r_3) > \lambda_{31_{1.05}}(r_3)$ as we show in Fig.~\ref{fig31b}.	
Finally, for the subinterval $ (1.44496, 2/b) $, we are considering the function $ \lambda_{11_{b}}(r_3) - \lambda_{31_{1.05}}(r_3) $, which is convex. To analyze it, we define the tangent lines at each endpoint of the interval; let $ L_{98}(r_3) $ and $ L_{99}(r_3) $ represent these lines, respectively. These tangent lines intersect at the point $ r_3 = 1.54371 $, where $ L_{98}(1.54371) = 0.00320258 $. Therefore the function $ \lambda_{11_{b}}(r_3) - \lambda_{31_{1.05}}(r_3) $ is strictly positive. See Fig.~\ref{fig32}.

\begin{figure}[ht]
 \centering
  \subfloat[Functions $\lambda_{11_{b}}(r_{3})$, $\lambda_{31_{1.05}}(r_{3})$, $L_{94}(r_{3})$ and $L_{95}(r_{3})$. 
  ]
  {
  \label{fig31a}
   \includegraphics[width=0.45\textwidth]{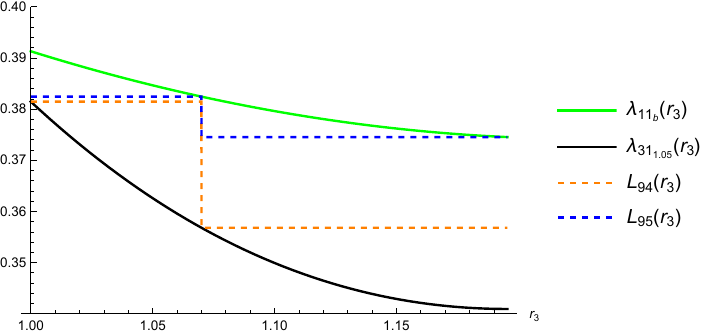}}
  \subfloat[Functions$\lambda_{11_{b}}(r_{3})$, $\lambda_{31_{1.05}}(r_{3})$, $L_{96}(r_{3})$ and $L_{97}(r_{3})$.
  ]
  {
   \label{fig31b}
    \ \ \includegraphics[width=0.45\textwidth]{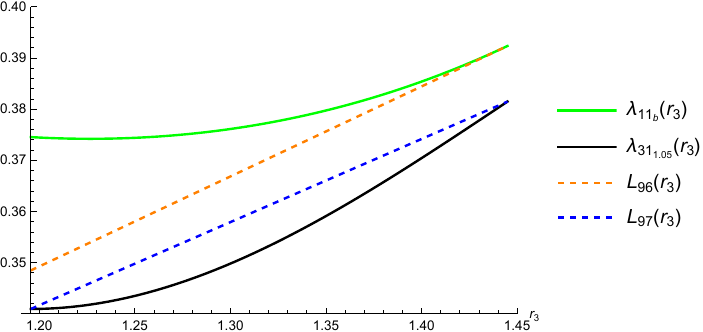}} 
   \caption{}
\end{figure}

\begin{figure}[ht]
	\centering
		\includegraphics[scale=0.8]{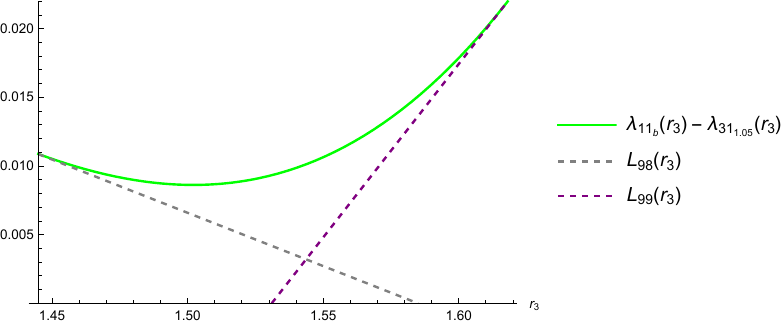}
		\caption{functions $\lambda_{11_{b}}(r_{3})-\lambda_{31_{1.05}}(r_{3})$, $L_{98}(r_{3})$ and $L_{99}(r_{3})$.}
		\label{fig32}
	\end{figure}
\end{proof}

\begin{proposition}\label{proposicionJ15}
For the region $J_{15}$, is true that $\lambda_{11_{\xi}}(r_{3})>\lambda_{52_{\xi}}(r_{3})$, $\xi\in[2.036,\infty)$.
\end{proposition}

\begin{proof}
First, we divide the interval $\xi$ in $[2,036,3)\cup[3,\infty)$.

\begin{itemize}
\item \underline{Part I}: for $\xi\in[2.036,3)$\,. The family of functions $ \lambda_{11_{\xi}}(r_3) $ is strictly decreasing, so we can consider the function $ \lambda_{11_3}(r_3) $, which is convex. We will construct the tangent line functions $ L_{123}(r_3) $, $ L_{124}(r_3) $, and $ L_{125}(r_3) $ at the points $ r_3 = \frac{2}{b} $, $ r_3 = 3 $, and $ r_3 = 3.47214 $, respectively.
The lines $ L_{123}(r_3) $ and $ L_{124}(r_3) $ intersect at $ r_3 = 2.58362 $, where $ L_{123}(2.58362) = 0.131495 $. Additionally, the lines $ L_{124}(r_3) $ and $ L_{125}(r_3) $ intersect at $ r_3 = 3.29344 $, with $ L_{124}(3.29344) = 0.606535 $.
On another note, we know that $ \alpha_5:=0.124054$ serves as an upper bound for the family $ \lambda_{52_{\xi}}(r_3) $. Hence, $ L_{123}(2.58362)>\alpha_{5}$ and  $ L_{124}(2.58362)>\alpha_{5}$. See Fig.~\ref{fig33}.

\begin{figure}[ht]
	\centering
		\includegraphics[scale=0.8]{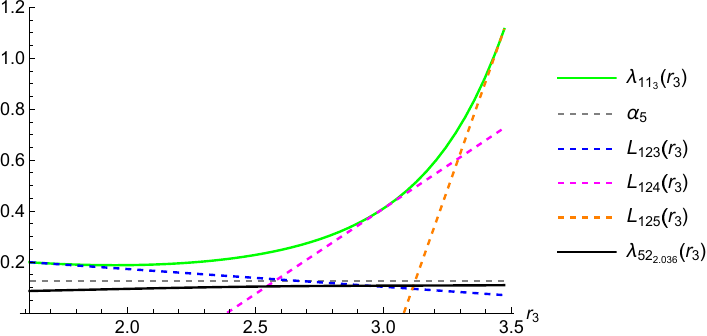}
		\caption{functions $\lambda_{11_{3}}(r_{3})$, $\alpha_{5}$, $L_{123}(r_{3})$, $L_{124}(r_{3})$, $L_{125}(r_{3})$ and $\lambda_{52_{2.036}}(r_{3})$.}
		\label{fig33}
	\end{figure}

\item \underline{Part II}: for $\xi\in[3,\infty)$\,. Since the family of functions $\lambda_{11_{3}}(r_{3})$ is strictly decreasing, we can define the function $ L_{127}(r_3) := \lim_{\xi \to \infty} \lambda_{11_{\xi}}(r_3)$, which is a function monotonically decreasing. Now, we consider the family $ L_{127}(r_3) - \lambda_{52_{\xi}}(r_3) $, which is a family strictly increasing and the functions of the family are monotonically decreasing. Hence, a lower bound for this family is $L_{127}((2/a) (1 + 3) + 1) - \lambda_{52_{3}}((2/a) (1 + 3) + 1)=0.00093769$.
Hence, we have $ \lambda_{11_{\xi}}(r_3) > \lambda_{52_{\xi}}(r_3) $. See Fig.~\ref{fig34}.
\begin{figure}[ht]
	\centering
		\includegraphics[scale=0.8]{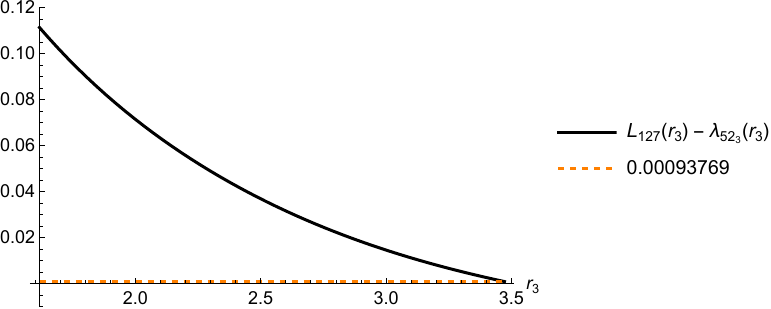}
		\caption{function $ L_{127}(r_3) - \lambda_{52_{3}}(r_3) $ and 0.00093769.}
		\label{fig34}
\end{figure}
\end{itemize}
\end{proof}

\begin{proposition}\label{proposicionJ16}
For the region $J_{16}$, the system \eqref{sistema} is not satisfied.
\end{proposition}
\begin{proof}
There exists a function $ \bar{r}_{5}(r_{3})$ such that $ \lambda_{21}(r_{3}, \bar{r}_{5}(r_{3})) = 0 $. Then, we divide the interval of $r_{3}$ into $(1,1.152781]\cup(1.152781, 1.201923]\cup (1.201923,1.3]$.  See Fig.~\ref{fig35}. 

\begin{figure}[ht]
	\centering
		\includegraphics[scale=0.8]{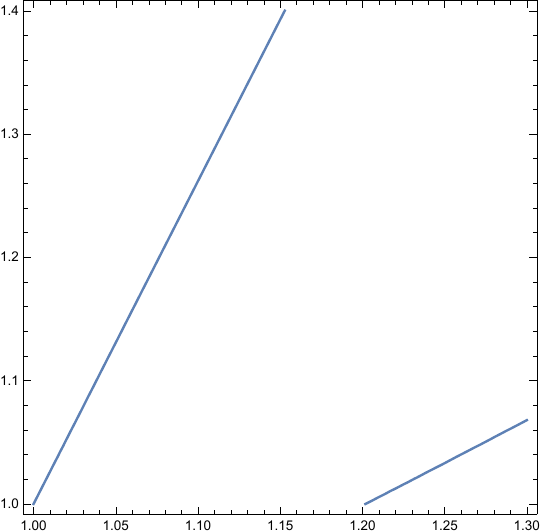}
		\caption{function $\lambda_{21}(r_{3},r_{5})$ for $r_{3}\in(1,1.3]$ and $r_{5}\in(1,1.4]$.}
		\label{fig35}
	\end{figure}
Is true that (See Figs.~\ref{fig36a},~\ref{fig36b} and~\ref{fig36c}),

\begin{equation*}
\begin{split}
\lambda_{21_{r_{3}}}(r_{5})&<\lambda_{41_{r_{3}}}(r_{5})\,, r_{3}\in (1, 1.13067],\\
\lambda_{21_{r_{3}}}(r_{5})&<\lambda_{11_{r_{3}}}(r_{5})\,, r_{3}\in (1.13067, 1.152781], \\
\lambda_{31_{r_{3}}}(r_{5})&<\lambda_{11_{r_{3}}}(r_{5})\,, r_{3}\in [1.201923,1.3]\,.
\end{split}
\end{equation*}

\begin{figure}[ht]
 \centering
  \subfloat[functions $\lambda_{21_{r_{3}}}(r_{5})$ and $\lambda_{41_{r_{3}}}(r_{5})$ for $r_{3}=1.008$. 
  ]
  {
  \label{fig36a}
   \includegraphics[width=0.45\textwidth]{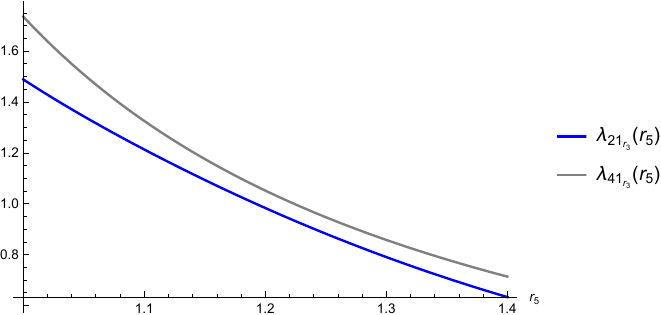}}
  \subfloat[functions $\lambda_{21_{r_{3}}}(r_{5})$ and $\lambda_{11_{r_{3}}}(r_{5})$ for $r_{3}=1.143$.
  ]
  {
   \label{fig36b}
    \ \ \includegraphics[width=0.45\textwidth]{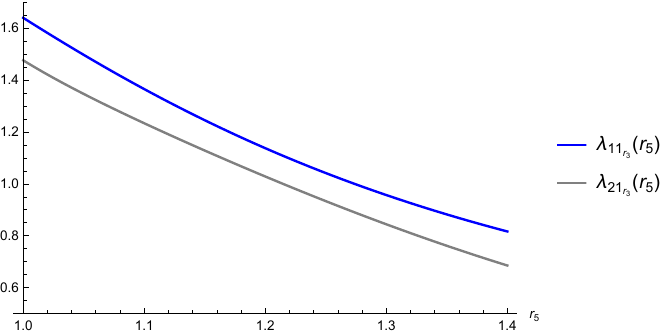}} \\

\subfloat[functions $\lambda_{11_{r_{3}}}(r_{5})$ and $\lambda_{31_{r_{3}}}(r_{5})$ for $r_{3}=1.2919$.
  ]
  {
  \label{fig36c}
   \includegraphics[width=0.45\textwidth]{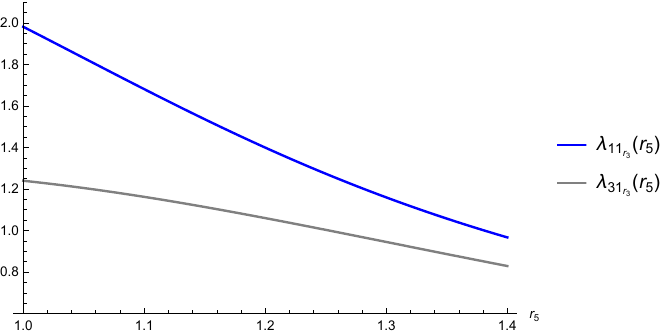}}

   \caption{}
\end{figure}
For the interval $(1.152781, 1.201923]$ the function $\lambda_{21_{r_{3}}}(r_{5})$ never vanishes.
\end{proof}	

\section{Conclusion}

The proof of Theorem \ref{th:existence} follows from the propositions discussed in Subsection \ref{subsec:Jn}. We have demonstrated the existence and uniqueness of star central configurations in the 5-body problem with equal masses. Specifically, we have shown that the only possible star central configuration corresponds to a regular pentagon. This finding confirms that, for $n \leq 5$, star configurations are limited to regular polygons, while for $n \geq 6$, this uniqueness is no longer applicable. These results lay the groundwork for future research on central configurations for larger values of $n$ and other types of symmetries.
%
%
%

\nocite{*}
\bibliography{aipsamp}%
\def\cprime{$'$} \def\lfhook#1{\setbox0=\hbox{#1}{\ooalign{\hidewidth
  \lower1.5ex\hbox{'}\hidewidth\crcr\unhbox0}}} \def\cprime{$'$}
  \def\cprime{$'$} \def\cprime{$'$} \def\cprime{$'$} \def\cprime{$'$}
  \def\cprime{$'$} \def\cprime{$'$}
\providecommand{\bysame}{\leavevmode\hbox to3em{\hrulefill}\thinspace}
\providecommand{\MR}{\relax\ifhmode\unskip\space\fi MR }
\providecommand{\MRhref}[2]{%
  \href{http://www.ams.org/mathscinet-getitem?mr=#1}{#2}
}
\providecommand{\href}[2]{#2}


\normalsize
\end{document}